\documentclass[12pt]{article}
\usepackage{amssymb,amsmath,amsthm,amsxtra,amsfonts}
\usepackage{enumerate}
\usepackage{natbib}
\usepackage{url} 

\usepackage{aliascnt}
\usepackage{color}
\usepackage[usenames,dvipsnames,table]{xcolor}
\usepackage{subcaption}
\usepackage{bbm}
\usepackage{dsfont}
\usepackage{lscape}
\usepackage{graphicx}
\usepackage{mathrsfs}
\usepackage{tikz}

\usepackage{algorithm,algpseudocode}
\usepackage{threeparttable}
\usepackage{booktabs}
\usepackage{rotating}
\usepackage{algorithm}
\usepackage{tabularx}
\usepackage{setspace}

\newcommand{\blind}{0}

\newcommand{\R}{\ensuremath{\mathbf{R}}}

\usepackage{longtable}
\usepackage{xr}

\theoremstyle{plain}
\newtheorem{theorem}{Theorem}[section]

\theoremstyle{definition}

\newtheorem{assumption}[theorem]{Assumption}

\theoremstyle{remark}
\newtheorem{remark}[theorem]{Remark}

\usepackage{xr}
\begin{document}

\def\spacingset#1{\renewcommand{\baselinestretch}%
{#1}\small\normalsize} \spacingset{1}


\if0\blind
{
	\title{\bf Nowcasting and aggregation: \\Why small Euro area countries matter}
	\author{Andrii Babii\\
		Department of Economics, University of North Carolina -- Chapel Hill\\
		and \\
		Luca Barbaglia \\
		European Commission, Joint Research Centre (JRC)\\
		and \\
		Eric Ghysels\\
		Department of Economics and Kenan-Flagler Business School, \\University of North Carolina -- Chapel Hill\\
		and \\
		Jonas Striaukas\\
		Department of Finance, Copenhagen Business School
	}
	\maketitle
} \fi

\if1\blind
{
	\title{\bf Nowcasting and aggregation: \\Why small Euro area countries matter}
\date{}
\maketitle
  \medskip
} \fi

\bigskip
\begin{abstract}
The paper studies the nowcasting of Euro area Gross Domestic Product (GDP) growth using mixed data sampling machine learning panel data regressions with both standard macro releases and daily news data.  
Using a panel of 19 Euro area countries, we investigate whether directly nowcasting the Euro area aggregate is better than weighted individual country nowcasts. 
Our results highlight the importance of the information from small- and medium-sized countries, particularly when including the COVID-19 pandemic period. 
The empirical analysis is supplemented by studying the so-called Big Four---France, Germany, Italy, and Spain---and the value added of news data when official statistics are lagging.
From a theoretical perspective, we formally show that the aggregation of individual components forecasted with pooled panel data regressions is superior to direct aggregate forecasting due to lower estimation error.
\end{abstract}

\noindent%
\vfill
\thispagestyle{empty}
\setcounter{page}{0}

\newpage
\spacingset{1.8} 

\section{Introduction}\label{sec:intro}

A much-researched example of nowcasting is that of US real GDP growth. Traditional methods rely on dynamic factor models which treat quarterly GDP growth as a latent process and use standard (monthly) macroeconomic data releases to obtain within-quarter estimates. There are several limitations of this approach. First, in the era of big data, it is challenging to expand dynamic factor models to include many high-frequency, non-traditional data increasingly used by macroeconomists to gauge the state of the economy. Second, going beyond the nowcasting of a single series to nowcasting the GDP growth for many countries simultaneously---or to nowcasting firms' earnings, which has many similarities with nowcasting GDP---is also challenging with the traditional models and methods. To that end, \cite{babii2022machine} developed machine learning mixed-data sampling (MIDAS) regressions for single series nowcasting using potentially large data sets, and \cite{babii2022machinepanel} extend these methods to machine learning panel data MIDAS regressions.

\medskip

In this paper, we are interested in nowcasting Euro area-19 (EA-19) GDP growth (see, e.g., \citealp{marcellino2003macroeconomic}), which means that we aim to potentially nowcast GDP growth for multiple countries simultaneously. One can think of several ways to proceed, namely:
\begin{itemize}
	\item[(a)] nowcast aggregate EA-19 GDP growth using only aggregate Euro area data;
	\item[(b)] nowcast aggregate EA-19 GDP growth using all available individual country data;
	\item[(c)] nowcast countries separately and then aggregate to obtain EA-19 nowcast;
	\item[(d)] nowcast only the large constituents, and use those to nowcast EA-19 GDP growth.
\end{itemize}

Option (a) is similar to the case of US GDP growth, with typically around 30 monthly series used to produce nowcasts. In particular, we use a larger set of monthly and weekly standard data for the EA-19 compared to the series used for the US, as will be explained in Section \ref{sec:standardmacro}. Even for this case, \cite{babii2022machine} show that machine learning methods outperform traditional dynamic factor models using exactly the same data. Moving to option (b), the challenge of dealing with large data sets emerges. If we keep the same 30 series, but collect them for each individual country, we have potentially 30 $\times$ 19 = 570 predictors. Still, the target is a single series, i.e.\ aggregate EA-19 GDP growth. The more interesting options are (c) and (d), which are the novel contributions of the paper. In both cases, the combination of nowcasting and aggregation comes into play. The former involves nowcasting for all countries, regardless of their size and importance in the overall economic outlook of the EA. Option (d) has been entertained by  \cite{cascaldigarcia2021} among others, who propose a multi-country nowcasting model to simultaneously predict the GDP of the Euro area and its three largest countries---namely Germany, France, and Italy---using up to 16 predictor variables per country.

\medskip

Cases (c) and (d) pertain to nowcasting in a data-rich environment using high-dimensional mixed-frequency panel data models. This is a relatively new and unexplored research area. \cite{khalaf2020dynamic} study low-dimensional mixed-frequency panels, but do not study forecasting. 
\cite{fosten2019panel} study nowcasting with a mixed-frequency vector autoregression (VAR) but not in the data-rich environment.
\cite{babii2022machinepanel} introduce structured machine learning panel data regressions, sampled at various frequencies using the sparse-group LASSO (sg-LASSO) regularization. They derive the oracle inequalities and debiased inference for this method in the panel data setting with dependent fat-tailed data; see also \cite{babii2022machine} and \cite{babii2020inference} for the time-series setting. \cite{babii2022panelpe} apply the method to nowcast a large cross-section of earnings data of US firms. In this paper, we apply the same framework for nowcasting panels of Euro area GDP growth based on real-time data vintages for standard macro data and news series extracted from a large set of newspaper articles. Our theoretical contribution pertains to the comparison of various aggregation schemes in high-dimensional data settings.

\medskip

The question of aggregate versus disaggregate estimation and forecasting has a long history in econometrics; see \cite{pesaran1989econometric} and references therein. \cite{lutkepohl2006forecasting} and \cite{hendry2011combining} discuss the theoretical comparison when the true data generating process is a VARMA process. Empirically, \cite{marcellino2003macroeconomic} find that aggregating individual time-series forecasts of EU inflation and GDP is superior to aggregate forecasting. However, the previous literature has not studied the aggregation problem for high-dimensional machine learning regressions, which is the focus of our work and is novel in the literature. We formally show that the aggregation of individual components forecasted with pooled panel data regressions is superior to direct aggregate forecasting due to lower estimation error. Importantly, our theoretical comparison is performed under weak assumptions on the underlying DGP, allowing for heterogeneous coefficients as in \cite{pesaran1989econometric}, \cite{pesaran1995estimating}, and \cite{pesaran2024forecasting}. Under heterogeneous coefficients, there is a trade-off between the overfitting of individual time-series regressions estimating a large number of coefficients and the misspecification error of pooled panel data regressions. The pooled panel data regressions may outperform the individual time-series regressions when the latter overfit due to large estimation error. Finally, simulations reported in the paper support the new theoretical results.

\medskip

In the empirical application, we use standard macroeconomic series and non-standard textual series. In recent years, the use of newspaper data for nowcasting has been explored by various authors. See, for example, \cite{thorsrud2018words}, \cite{larsen2021news}, \cite{ellingsen2022news}, 
\cite{barbaglia2022forecastingUS}, and 
\cite{boss2025nowcasting}
among others. 
\cite{SCOTTI20161}, \cite{baker2016}, \cite{barbaglia2022forecastingEA}, and \cite{ashwin2021nowcasting} are recent examples that use newspaper data for economic analysis in multiple countries. 
Although the construction of country-specific indicators brings an additional level of complexity, the presumption is that the resulting indices provide an early signal of the local economic conditions that is more precise than a news-based general indicator.
We follow \cite{barbaglia2022forecastingEA}, who propose country-specific text-based indicators for five European countries using local news translated to English. In this paper, we extend their analysis to all EA-19 countries and show how the inclusion of news-based indices about smaller countries can improve the nowcasting performance.

\medskip

We use the MIDAS sparse-group LASSO regression approach of \cite{babii2022machine}, suitable for large-dimensional data environments, to test whether adding different levels of heterogeneity to panel models improves the quality of nowcasts. 
Moreover, we test whether specific data sources are informative when used in nowcasting settings. In particular, we look at whether news data can improve nowcasts over more traditionally used macro and financial series. 
Lastly, we apply several weighting schemes to compute the Euro area aggregate nowcast based on country-level GDP nowcasts. 
We show that this strategy leads to more accurate nowcasts irrespective of the weighting scheme, suggesting that information in all European countries contributes to overall predictions. 
In addition, we test whether smaller panels that include only the Big Four countries---France, Germany, Italy, and Spain---suffice to compute the Euro area nowcasts. 
Our results suggest that smaller countries indeed matter and models that incorporate all EA-19 produce higher quality nowcasts.  
The efficient estimation of panel models using large data sets benefits from the timely inclusion of information about the current state of the economy from smaller countries. 

\medskip

In sum, the paper makes multiple contributions. First, we extend the existing literature on forecasting/nowcasting and aggregation to high-dimensional data settings. Second, we study different aggregation schemes and find that a projection method appears to work best. Third, we showcase the use of non-traditional data in nowcasting, and more specifically news-related data, a topic that has received much attention recently. Fourth, we highlight the role played by the small Euro area countries due to their timely data releases and interconnectedness of their economies with the larger countries. Finally, our sample includes the pandemic, and we analyze its impact on nowcasting and model performance. The latter leads us to investigate nowcast combination schemes across different models.

\medskip

The paper is organized as follows. Section \ref{sec:theory} presents the theoretical comparison of various aggregation schemes in high-dimensional data settings. Section \ref{sec:simul} presents the results of a Monte Carlo simulation study. Section \ref{sec:eunowcasting} considers the empirical problem of nowcasting EU output and aggregation. Finally, Section \ref{sec:conclusion} concludes. The Appendix contains additional empirical results and details on the construction of the news-based indicators.


\section{Forecasting/Nowcasting and Aggregation in High-Dimensional Data Settings \label{sec:theory}}

The theoretical developments in this section use a standard linear panel data model setting used for the purpose of forecasting. Hence, we simplify the notation by ignoring the mixed frequency nature of the data and treat nowcasting as a special case of a generic forecasting problem. More specifically, we work with panel data $(y_{i,t},x_{i,t})\in\R\times\R^p$, where $t=1,2,\dots,T$ denotes time and $i=1,2,\dots,N$ is the cross-section. The vector of covariates $x_{i,t}$ can include $1$ for the intercept, the lags of $y_{i,t}$, and the (higher-frequency observations of) covariates. The objective is to predict the aggregate outcome:
\begin{equation*}
	Y_{t+1} := \omega^\top\mathbf{y}_{t+1} = \sum_{i=1}^N\omega_iy_{i,t+1}
\end{equation*}
using the covariates $\mathbf{x}_{t}\in\R^{N\times p}$, where $\mathbf{y}_t = (y_{1,t},\dots,y_{N,t})^\top$ and $\omega\in\R^N$ is a vector of aggregation weights. The simplest aggregation scheme is averaging, in which case the weights are $\omega_i=1/N$ for all units $i=1,\dots,N$. 
For a quadratic loss function, the optimal forecasts are given by conditional expectations, hence: $\mathbb{E}[Y_{t+1}|\mathcal{F}_t]$ = $\sum_{i=1}^N\omega_i\mathbb{E}[y_{i,t+1}|\mathcal{F}_t],$
where $\mathcal{F}_t$ is the information available at time $t$. Formally, $\mathcal{F}_t$ is the $\sigma$-field generated by $(y_{i,s},x_{i,s})$ for $i=1,\dots,N$ and $s\leq t$. This allows the disaggregate series $x_{i,t},i=1,\dots,N$ to be used when forecasting the aggregate $Y_{t+1}$. If one had perfect knowledge of the data-generating process (DGP), forecasting the aggregate outcome would be equivalent to linear aggregation of individual forecasts. 	However, in practice forecasts are estimated using data and the interplay between the estimation and specification errors becomes important.

\medskip

We focus on linear forecasting models with one of the following two approaches:
\begin{enumerate}
	\item \textbf{Direct aggregate forecasting.} Two possibilities arise here:
	\begin{enumerate}
		\item We regress the aggregate outcome $Y_{t+1}$ on the aggregate covariates computed as $\omega^\top\mathbf{x}_t$ = $\sum_{i=1}^N\omega_ix_{i,t}.$
		If $\hat\beta^{\rm A}\in\R^p$ are the estimated regression coefficients, the forecasted aggregate outcome is
			$\hat Y_{t+1}^{\rm A}$ = $\sum_{i=1}^N\omega_ix_{i,t}^\top\hat\beta^{\rm A}.$
		\item We regress the aggregate outcome $Y_{t+1}$ on all individual country-level covariates $\mathbf{x}_t = (x_{1,t},\dots,x_{N,t})^\top$ weighted by $\omega\in\R^N$. If $\hat\beta_i^{\rm AC}\in\R^p$ are the estimated regression coefficients for country $i$, the forecasted aggregate outcome is
			$\hat Y_{t+1}^{\rm AC}$ = $\sum_{i=1}^N\omega_ix_{i,t}^\top\hat\beta_i^{\rm AC}.$
	\end{enumerate}
	\item \textbf{Individual component forecasting.} This approach consists of forecasting the entire panel of individual outcomes. Two possibilities arise here:
	\begin{enumerate}
		\item individual time-series regressions with predictions $\hat y_{i,t+1}^{\rm TS}=x_{i,t}^\top\hat\beta_i^{\rm TS}$ allowing for heterogeneous coefficients;
		\item pooled panel data regressions with predictions $\hat y_{i,t+1}^{\rm P}=x_{i,t}^\top\hat\beta^{\rm P}$ having homogeneous coefficients.
	\end{enumerate}
	The forecasted aggregate outcome in both cases is
		$\hat Y_{t+1}^{k}$ = $\sum_{i=1}^N\omega_i\hat y_{i,t+1}^{\rm k},$ $k\in\{\mathrm{TS}, \mathrm{P}\}.$
\end{enumerate}
It is assumed that the DGP is as follows:
\begin{equation}\label{eq:dgp}
	y_{i,t+1} = x_{i,t}^\top\beta_{i} + u_{i,t+1}, 	\qquad	\mathbb{E}[u_{i,t+1}\mid\mathcal{F}_t]=0
\end{equation}
with $i=1,\dots,N, \ t=1,\dots,T$. For $k\in\{\mathrm{A,AC,TS,P}\}$, let
$\hat Y_{T+1}^k$ = $\sum_{i=1}^{N}\omega_i x_{i,T}^{\top}\hat\beta_i^k,$ and 
$Y_{T+1}$ = $\sum_{i=1}^{N}\omega_i y_{i,T+1}$
be the forecasted and the realized aggregate outcomes. Using the above equation, we can write the forecast error as a sum of estimation error, heterogeneous coefficients bias, and irreducible error:
\begin{equation}\label{eq:forecast_error}
		\hat Y_{T+1}^k - Y_{T+1}  = \sum_{i=1}^{N}\omega_i x_{i,T}^{\top}(\hat\beta_i^k  - \beta_i^k) + \underbrace{\sum_{i=1}^{N}\omega_i x_{i,T}^{\top}(\beta_i^k  - \beta_i)}_{=:H_{N,T}^k} -  \sum_{i=1}^{N}\omega_i u_{i,T+1},
\end{equation}
where $\beta_i^k$ is a vector of model-specific linear projection coefficients and $H_{N,T}^k$ is the heterogeneity bias.
Combining the identity in Equation~\eqref{eq:forecast_error} with $\mathbb{E}[u_{i,T+1}\mid\mathcal{F}_T]=0$ in Equation~\eqref{eq:dgp}, the mean-squared forecasting error is decomposed as:
\begin{eqnarray}\label{eq:msfe}
		\mathrm{MSFE}(\hat Y_{T+1}^k) & = &  \left|\sum_{i=1}^{N}\omega_i x_{i,T}^{\top}(\hat\beta_i^k  - \beta_i^k) + H_{N,T}^k\right|^2 + \mathrm{Var}\left(\sum_{i=1}^{N}\omega_i u_{i,T+1}|\mathcal{F}_T\right) \\
		& \leq &   2\left|\sum_{i=1}^{N}\omega_i x_{i,T}^{\top}(\hat\beta_i^k  - \beta_i^k)\right|^2 + 2\left|H_{N,T}^{k}\right|^2 + \omega^\top\Sigma_h\omega, \notag
\end{eqnarray}
where $\Sigma_h = \mathrm{Var}(\mathbf{u}_{T+1}|\mathcal{F}_T)$. In the aggregate-on-aggregate regressions the aggregate outcome is regressed on aggregate covariates to obtain the direct aggregate forecast. Let $\hat\beta^{\rm A}$ be the estimated slope coefficient and let
	$\beta^{\rm A}$ :=  $\arg\min_{b\in\R^p}\mathbb{E}\left|Y_{t+1} - \sum_{i=1}^N\omega_ix_{i,t}^\top b\right|^2$
be the corresponding population projection coefficients. The forecasted aggregate outcome is $\hat Y_{T+1}^{\rm A}=\sum_{i=1}^N\omega_ix_{i,T}^\top\hat\beta^{\rm A}$.

In what follows, we focus on estimators that correspond either to individual time-series regressions or to pooled panel data regressions with sparse-group LASSO regularization.
\begin{assumption}\label{as:subG_indep_weights}
	(i) the data are sub-Gaussian, stationary, exponentially $\tau$-mixing, and cross-sectionally independent; (ii) the restricted eigenvalue condition holds; (iii) the LASSO penalty is $\lambda\sim\sqrt{\log(p)/T}$; (iv) the approximation error is negligible; (v) the weights are such that $|\omega|_\infty = O(N^{-1})$.
\end{assumption}
\noindent Assumption~\ref{as:subG_indep_weights} (i)-(iv) corresponds to \cite{babii2022machine}, Assumptions 3.1-3.4, where a technically precise statement can be found. The assumption on weights holds for the averaging aggregation scheme.

\begin{theorem}\label{thm:agg_on_agg}
	Under Assumption~\ref{as:subG_indep_weights}, the mean-squared forecasting error of the aggregate-on-aggregate regression forecast is
		$\mathrm{MSFE}(\hat Y_{T+1}^{\rm A})$ = $O_P\left(\frac{s^2\log^2 p}{NT} + \left|H_{N,T}^{\rm A}\right|^2 + \frac{1}{N}\right),$
	where $s=\mathrm{supp}(\beta^{\rm A})$ is the sparsity of the population projection coefficient.
\end{theorem}
\begin{proof}
	Under Assumption~\ref{as:subG_indep_weights}, we obtain from \cite{babii2022machine} that
		$|\hat\beta^{\rm A} - \beta^{\rm A}|_1$ = $O_P\left(s\sqrt{\frac{\log p}{T}}\right).$
	Strictly speaking, \cite{babii2022machine} assume polynomial tails/mixing which leads to an additional polynomial term vanishing under exponential mixing/tails; see also \cite{merlevede2011bernstein}, Theorem 1. If errors are cross-sectionally independent, then $\Sigma_h$ is diagonal and $\omega^\top\Sigma_h\omega = O(N^{-1})$ under Assumption~\ref{as:subG_indep_weights} (v). Similarly, under Assumption~\ref{as:subG_indep_weights} (i) and (v), we have
		$\max_{j\in[p]}\mathrm{Var}\left(\sum_{i=1}^N\omega_ix_{i,T,j}\right)$ = $O(N^{-1}),$
	so that
		$\max_{j\in[p]}\left|\sum_{i=1}^N\omega_ix_{i,t,j}\right|$ = $O_P\left(\sqrt{\frac{\log p}{N}}\right).$
		Combining the above computations with the inequality in Equation (\ref{eq:msfe}) and H\"{o}lder's inequality, we obtain
	\begin{equation*}
		\begin{aligned}
			\mathrm{MSFE}(\hat Y_{T+1}^{\rm A}) & = \mathbb{E}\left[(\hat Y_{T+1}^{\rm A} - Y_{T+1})^{2}\mid\mathcal{F}_{T}\right] \\
			& \leq 2\left|\sum_{i=1}^N\omega_ix_{i,T}^\top(\hat\beta^{\rm A} - \beta^{\rm A})\right|^2 + 2\left|H_{N,T}^{\rm A}\right|^2 + \omega^\top\Sigma_h\omega \\
			& \leq 2|\hat\beta^{\rm A} - \beta^{\rm A}|_1^2\max_{j\in[p]}\left|\sum_{i=1}^N\omega_ix_{i,T,j}\right|^2 + O_P\left(\left|H_{N,T}^{\rm A}\right|^2 + \frac{1}{N}\right) \\
			& = O_P\left(\frac{s^2\log^2 p}{NT} + \left|H_{N,T}^{\rm A}\right|^2 + \frac{1}{N}\right).
		\end{aligned}
	\end{equation*}
\end{proof}

\begin{remark}[Homogeneous Slopes]
	It is easy to see that when the true coefficient is homogeneous, i.e., $\beta_i=\beta$ for all $i=1,\dots,N$, then the population projection coefficient is $\beta^{\rm A}=\beta$ and the heterogeneity bias  vanishes; see Equation~\eqref{eq:forecast_error}. In this case, we obtain
$\mathrm{MSFE}(\hat Y_{T+1}^{\rm A})$ = $O_P\left(\frac{s^2\log^2 p}{NT} + \frac{1}{N}\right).$
\end{remark}

\begin{remark}[Why $N$ helps?]
	The quality of forecasts improves with larger cross-sectional size $N$ since the irreducible error is of order $O(N^{-1})$. In contrast, the irreducible error never vanishes with time series data alone.
\end{remark}

\paragraph{Aggregate-on-Components Regression:} In the aggregate-on-components regressions the aggregate outcome is regressed on all individual component covariates. Let $\hat\beta^{\rm AC}\in\R^{pN}$ be the vector of estimated coefficients and
	$\beta^{\rm AC}$ = $\arg\min_{\{b_i\}_{i=1}^N}\mathbb{E}\left|Y_{t+1} - \sum_{i=1}^N\omega_ix_{i,t}^\top b_i\right|^2$
be the corresponding population projection coefficients. The forecasted aggregate outcome is $\hat Y_{T+1}^{\rm AC} = \sum_{i=1}^N\omega_ix_{i,T}^\top\hat\beta_i^{\rm AC}$. 

\begin{theorem}\label{thm:agg_on_comp}
	Under Assumption~\ref{as:subG_indep_weights}, the mean-squared forecasting error of the aggregate-on-components regression forecast is
		$\mathrm{MSFE}(\hat Y_{T+1}^{\rm AC})$ = $O_P\left(\frac{s_{\max}^2\log^2(pN)}{T} + \frac{1}{N}\right),$
	where $s_{\max}=\max_{i\in[N]}\mathrm{supp}(\beta_i^{\rm AC})$ is the maximum sparsity of the population projection coefficients.
\end{theorem}
\begin{proof}
	Note that since $Y_{t+1} = \sum_{i=1}^N\omega_ix_{i,t}^\top\beta_i + \sum_{i=1}^N\omega_iu_{i,t+1}$, see Equation~\eqref{eq:dgp}, we have $\beta_i^{\rm AC} = \beta_i$ for all $i=1,\dots,N$. This shows that the heterogeneity bias is zero: $H_{N,T}^{\rm AC}=0$. Note also that under Assumption~\ref{as:subG_indep_weights}, we obtain from \cite{babii2022machine}: $|\hat\beta^{\rm AC} - \beta^{\rm AC}|_1$ = $O_P\left(s_{\max}N\sqrt{\frac{\log (pN)}{T}} \right).$
	Under Assumption~\ref{as:subG_indep_weights} (i) and (v), we also have
		$\max_{i\in[N],j\in[p]}\left|\omega_ix_{i,T,j}\right|$ $\leq$ $|\omega|_\infty\max_{i\in[N],j\in[p]}|x_{i,T,j}|$ = $O_P\left(\frac{\sqrt{\log(pN)}}{N}\right).$
	Combining these computations with Equation (\ref{eq:msfe}), the mean-squared forecasting error is
	\begin{equation*}
		\begin{aligned}
			\mathrm{MSFE}(\hat Y_{T+1}^{\rm AC}) & = \mathbb{E}\left[(\hat Y_{T+1}^{\rm AC} - Y_{T+1})^{2}\mid\mathcal{F}_{T}\right] \\
			& =  \left|\sum_{i=1}^{N}\omega_i x_{i,T}^{\top}(\hat\beta_i^{\rm AC}  - \beta_i^{\rm AC})\right|^2 + \omega^\top\Sigma_h\omega \\
			& \leq \max_{i\in[N],j\in[p]}\left|\omega_ix_{i,T,j}\right|^2|\hat\beta^{\rm AC} - \beta^{\rm AC}|_1^2 + \omega^\top\Sigma_h\omega  = O_P\left(\frac{s_{\max}^2\log^2(pN)}{T} + \frac{1}{N}\right).
		\end{aligned}
	\end{equation*}
\end{proof}

\subsection{Individual Component Forecasting}	
\paragraph{Pooled Panel Regressions:}
In the pooled panel data regressions, the entire panel is used to estimate a single homogeneous regression. Let $\hat\beta^{\rm P}\in\R^p$ be the vector of estimated coefficients and let $\beta^{\rm P}$ = $\arg\min_{b\in\R^p}\mathbb{E}\left|y_{i,t+1}-x_{i,t}^\top b\right|^2$
be the corresponding population projection coefficient. The forecasted aggregate outcome is $\hat Y_{T+1}^{\rm P}$ = $\sum_{i=1}^N\omega_ix_{i,T}^\top\hat\beta^{\rm P}$.

\begin{theorem}\label{thm:pool}
	Suppose that Assumption~\ref{as:subG_indep_weights} (i)-(ii) and (iv)-(v) hold. If $\lambda\sim \sqrt{\log(p)/(NT)}$, then
		$\mathrm{MSFE}(\hat Y_{T+1}^{\rm P})$ = $O_P\left(s^2\frac{\log^2 p}{N^2T} + \left|H_{N,T}^{\rm P}\right|^2 + \frac{1}{N}\right),$
	where $s=\mathrm{supp}(\beta^{\rm P})$ is the sparsity of the population projection coefficient.
\end{theorem}
\begin{proof}
	Under Assumption~\ref{as:subG_indep_weights}, we obtain from \cite{babii2022machinepanel} that
		$|\hat\beta^{\rm P} - \beta^{\rm P}|_1$ = $O_P\left(s\sqrt{\frac{\log p}{NT}} \right).$
	Combining these computations with Equation (\ref{eq:msfe}), the mean-squared forecasting error is
	\begin{equation*}
		\begin{aligned}
			\mathrm{MSFE}(\hat Y_{T+1}^{\rm P}) & = \mathbb{E}\left[(\hat Y_{T+1}^{\rm P} - Y_{T+1})^{2}\mid\mathcal{F}_{T}\right]  = \left|\sum_{i=1}^N\omega_ix_{i,T}^\top(\hat\beta^{\rm P} - \beta^{\rm P})\right|^2 + \omega^\top\Sigma_h\omega \\
			& \leq 2|\hat\beta^{\rm P} - \beta^{\rm P}|_1^2\max_{j\in[p]}\left|\sum_{i=1}^N\omega_ix_{i,T,j}\right|^2 + O_P\left(\left|H_{N,T}^{\rm P}\right|^2 + \frac{1}{N}\right) \\
			& = O_P\left(s^2\frac{\log^2 p}{N^2T} + \left|H_{N,T}^{\rm P}\right|^2 + \frac{1}{N}\right).
		\end{aligned}
	\end{equation*}
\end{proof}

\begin{remark}
	The pooled panel data regression forecast is always better than the aggregate-on-aggregate regression forecast. It has lower estimation error than the aggregate-on-components regression forecast, but the precise ranking of these two methods depends on whether heterogeneity bias is substantial relative to the estimation error.
\end{remark}

\paragraph{Individual Time-Series Regressions:}
In the individual time-series regressions, each individual time series is forecasted separately. Let $\hat\beta_i^{\rm TS}\in\R^p$ be the vector of estimated coefficients for unit $i$ and let $\beta_i^{\rm TS}$ = $\arg\min_{b\in\R^p}\mathbb{E}\left|y_{i,t+1}-x_{i,t}^\top b\right|^2$ be the corresponding population projection coefficient. The forecasted aggregate outcome is $\hat Y_{T+1}^{\rm TS} = \sum_{i=1}^N\omega_ix_{i,T}^\top\hat\beta_i^{\rm TS}$.

\begin{theorem}\label{thm:indTS}
	Under Assumption~\ref{as:subG_indep_weights}, the mean-squared forecasting error of the individual time-series regression forecast is
		$\mathrm{MSFE}(\hat Y_{T+1}^{\rm TS})$ = $O_P\left(\sum_{i=1}^Ns_i^2\frac{\log(pN)\log p}{NT} + \frac{1}{N}\right),$
	where $s_i=\mathrm{supp}(\beta_i^{\rm TS})$ is the sparsity of the population projection coefficient for unit $i$.
\end{theorem}
\begin{proof}
	Under Assumption~\ref{as:subG_indep_weights}, by \cite{babii2022machine}, we have $|\hat\beta_i^{\rm TS} - \beta_i^{\rm TS}|_1$ = $O_P\left(s_i\sqrt{\frac{\log p}{T}}\right).$ Combining these computations with Equation (\ref{eq:msfe}) and H\"{o}lder's inequality, we obtain
	\begin{equation*}
		\begin{aligned}
			\mathrm{MSFE}(\hat Y_{T+1}^{\rm TS}) & =\mathbb{E}\left[(\hat Y_{T+1}^{\rm TS} - Y_{T+1})^{2}\mid\mathcal{F}_{T}\right]  =  \left|\sum_{i=1}^{N}\omega_i x_{i,T}^{\top}(\hat\beta_i^{\rm TS}  - \beta_i^{\rm TS})\right|^2 + \omega^\top\Sigma_h\omega \\
			& \leq \sum_{i=1}^N\omega_i^2\sum_{i=1}^N\left|x_{i,T}^\top(\hat\beta_i^{\rm TS} - \beta_i^{\rm TS})\right|^2 + O_P(N^{-1}) \\
			& \leq \max_{1\leq i\leq N}|x_{i,T}|_\infty^2\sum_{i=1}^N\omega_i^2\sum_{i=1}^N\left|\hat\beta_i^{\rm TS} - \beta_i^{\rm TS}\right|_1^2 + O_P(N^{-1}) \\
			& = O_P\left(\sum_{i=1}^Ns_i^2\frac{\log(pN)\log p}{NT} + \frac{1}{N}\right).
		\end{aligned}
	\end{equation*}
\end{proof}

\begin{remark}
	The individual time-series regression forecast is slightly better than the aggregate-on-components regression forecast. It has higher estimation error than the pooled panel data regression forecast, but the precise ranking of these two methods depends on whether heterogeneity bias is substantial relative to the estimation error.
\end{remark}

\begin{remark}
	 To summarize the above results, we use $\mathrm{MSFE}(\hat Y_{T+1}^{\rm a})$ $\gtrsim$ $\mathrm{MSFE}(\hat Y_{T+1}^{\rm b})$ to denote that the convergence rate of $\mathrm{MSFE}(\hat Y_{T+1}^{\rm a})$ is slower than that of $\mathrm{MSFE}(\hat Y_{T+1}^{\rm b})$ with $a,b\in\{\mathrm{TS},\mathrm{AC},\mathrm{P},\mathrm{A}\}$. We have the following ranking of the mean-squared forecasting errors:
	\begin{enumerate}
		\item We always have:
			$\mathrm{MSFE}(\hat Y_{T+1}^{\rm TS})$ $\gtrsim$ $\mathrm{MSFE}(\hat Y_{T+1}^{\rm AC})$ and $\mathrm{MSFE}(\hat Y_{T+1}^{\rm P})$ $\gtrsim$ $\mathrm{MSFE}(\hat Y_{T+1}^{\rm A})$.
		\item With homogeneous coefficients: $\mathrm{MSFE}(\hat Y_{T+1}^{\rm P})$ $\gtrsim$ $\mathrm{MSFE}(\hat Y_{T+1}^{\rm A})$ $\gtrsim$ $\mathrm{MSFE}(\hat Y_{T+1}^{\rm TS})$ $\gtrsim$ $\mathrm{MSFE}(\hat Y_{T+1}^{\rm AC}).$
		\item With heterogeneous coefficients:
			$\mathrm{MSFE}(\hat Y_{T+1}^{\rm TS})$ $\gtrsim$ $\mathrm{MSFE}(\hat Y_{T+1}^{\rm AC})$ $\gtrsim$ $\mathrm{MSFE}(\hat Y_{T+1}^{\rm P}) \gtrsim \mathrm{MSFE}(\hat Y_{T+1}^{\rm A}),$
		provided that heterogeneity bias is non-negligible.
	\end{enumerate}
\end{remark}

\section{Simulation Study \label{sec:simul}}

\paragraph{Data generating process:}
In the simulation study, we consider the panel DGPs for the number of units (e.g., countries) $i=1,\dots,N,$ and time units $t=1,\dots,T$, where $N\in\{10,20\}$ and $T\in\{35, 100\}$, respectively. We generate $j=1,\dots,p$ regressors, $p\in\{50,500\}$. In our setup, the signal is $j=1$ while $j\ge2$ are noisy regressors. To simplify the DGP, we only consider the same frequency regressors as the outcome variable, hence we apply the LASSO estimator. The model we simulate from has the following form
\begin{equation*}
	y_{i,t} \;=\; \alpha_i \;+\; \gamma_i\, y_{i,t-1} \;+\; \beta_i\, x_{i,t,1} \;+\; \varepsilon_{i,t}.
\end{equation*}
For each $j\in[p]$, the regressor $x_{i,t,j}$ is an AR(1) process simulated  around unit-specific locations, i.e., 
$x_{i,t,j}$ =  $\mu_{i,j} + \xi_{i,t,j},$ and  
$\xi_{i,t,j}$ = $\rho_{i,j}\,\xi_{i,t-1,j} + \sigma_{i,j}\sqrt{1-\rho_{i,j}^2}\,\nu_{i,t,j},$
with $\nu_{i,t,j}\sim\mathcal N(0,1)$ for the Gaussian design and $\nu_{i,t,j}\sim\text{student}$-$t(5)$ for the heavy-tailed case, $\rho_{i,j}\sim\text{Uniform}(0,0.95)$, $\sigma_{i,j}>0$. We simulate $\sigma_{i,j}^2 \sim (1 + \chi^2_1)/2$. Lastly,  $ \mu_{i,j}  = (z_{i,j}^2 - 1) / \sqrt{2}$ where $ z_{i,j}\sim \mathcal N(0,1)$ or $ z_{i,j}\sim\text{student}$-$t(5)$.

\paragraph{Heterogeneity:}
First, we impose a homogeneous AR coefficient, i.e, we set it to  $\gamma_i = 0.688,\forall i\in[N].$
The heterogeneous coefficients are the intercept $\alpha_i$ and the slope $\beta_i$.  We assume the following structure for $\alpha_i$ and $\beta_i$, respectively, 
$\alpha_i$ = $\alpha_0$ + $\phi\,\mu_{i,1}$ + $\sigma_\eta\,\eta_i,$ and 
$\beta_i$ = $\beta_0$ + $\pi\,\mu_{i,1}$ + $\sigma_\zeta\,\zeta_i,$
with $\eta_i,\zeta_i\sim\mathcal N(0,1)$ for Gaussian or $\eta_i,\zeta_i\sim\text{student}$-$t(5)$ for heavy-tailed scenarios, and $\phi=\rho_{\alpha x}\mathbf{\sigma}$, $\pi=\rho_{\beta x}\mathbf{\sigma}$, 
$\sigma_\eta^2=\sigma^2-\phi^2$, $\sigma_\zeta^2=\sigma^2-\pi^2$ and $\alpha_0=0$, $\beta_0=0.5$. Lastly, the error term is simulated from 
$\varepsilon_{it}$ = $\sigma_{\varepsilon,i}\,(\zeta_{it}^2-1)/\sqrt{2},$ 
with $\zeta_{it}\sim\mathcal N(0,1)$ or $\sim\text{student}\text{-}t(5),$  $\sigma_{\varepsilon,i}>0.$
The other free parameters are fixed to the following values: $\rho_{\alpha x}=0.1$, $\rho_{\beta x}=0.5$. Throughout the experiment, we vary the parameter $\sigma\in\{0, 0.2, 0.4, 0.6, 0.8\}$, which controls the level of heterogeneity. For $\sigma=0$, $\alpha_i = \alpha_0 = 0$ and $\beta_i = \beta_0 = 0.5$, while $\sigma\neq0$ implies heterogeneous $\alpha_i$ and $\beta_i$.

\medskip

We compute predictions using four models which we estimate by applying the LASSO estimator. The first is the pooled regression model for all $N$ units, which we denote as P. Next, we estimate individual regressions for each $i$, denoted TS. The third case considers a model where we predict the aggregate outcome using all $N$ individual regressors. We call this model AC. Lastly, we consider aggregate-on-aggregate approach, denoted A, where we regress the aggregate outcome on aggregate regressors. Each aggregation assumes fixed weights set to $1/N$. 

\medskip

In Table \ref{tab:simul} we report results for the individual time-series (TS), aggregate-on-components (AC), and aggregate-on-aggregate (A) approaches, all relative to the pooled model (P) out-of-sample mean squared errors (MSE). We compute MSEs for all approaches for the target out-of-sample aggregate outcome. For pooled and individual cases, we use $1/N$ weights to aggregate our predictions. 

\medskip

The Monte Carlo simulations show that pooled regressions strike a balance between estimation error and heterogeneity bias. Compared to aggregate-on-aggregate and aggregate-on-components regressions, the pooled model generally achieves lower mean squared errors because it benefits from the larger effective sample size, which reduces estimation noise. However, when heterogeneity across units is strong, pooling imposes a homogeneity restriction that creates bias. This trade-off explains why pooled regressions outperform individual time-series regressions when the latter overfit in small samples, but can be worse when unit-specific variation matters. In short, pooled regressions dominate aggregate-based methods and are more robust than TS in finite samples, yet they may underperform when heterogeneity is substantial. 

\medskip

The simulation results also reveal that pooled regressions behave differently under distributional and dimensional changes. Moving from Gaussian to student-$t(5)$ innovations, the pooled model remains relatively stable, while competing methods such as TS and AC become more sensitive to heavy tails and deteriorate in accuracy. This robustness highlights the pooled model’s advantage in heavy-tailed environments where heavy-tailed observations inflate estimation errors elsewhere. Increasing dimensionality from $p=50$ to $p=500$ introduces additional noise, which tends to amplify the overfitting of TS and AC methods, whereas pooling, again leveraging the combined sample size, better controls variance. Overall, across both the distributional shift from Gaussian to student-$t(5)$ and the dimensional increase, pooled regressions exhibit the most reliable performance relative to the other three approaches and the results are in line with our theoretical results.
\begin{center}
	\begin{table}[ht]
			{\footnotesize
		\begin{tabular}{rccccccccccccc}
			& & & & & & & & & & & & & \\
			& \multicolumn{6}{c}{Panel A: Gaussian design} & & \multicolumn{6}{c}{Panel B: student-$t(5)$ design} \\
			&\multicolumn{3}{c}{p = 50}&\multicolumn{3}{c}{p = 500} & & \multicolumn{3}{c}{p = 50}&\multicolumn{3}{c}{p = 500} \\
			N$\backslash$T& TS & AC & A & TS & AC & A & & TS & AC & A & TS & AC & A\\ 
			\hline
			&\multicolumn{13}{c}{$\sigma$ = 0}\\
			10$\backslash$35 & 1.608 & 1.561 & 1.907 & 1.711 & 1.676 & 1.630 & & 1.748 & 1.529 & 1.815 & 1.891 & 1.886 & 1.623 \\ 
			20$\backslash$35 & 1.524 & 1.627 & 1.890 & 1.684 & 1.656 & 1.618 & & 1.719 & 1.471 & 1.740 & 1.857 & 1.828 & 1.609 \\ 
			10$\backslash$100 & 1.095 & 1.183 & 1.887 & 1.272 & 1.244 & 1.731 & & 1.324 & 1.145 & 1.753 & 1.407 & 1.445 & 1.694 \\ 
			20$\backslash$100 & 1.144 & 1.097 & 1.726 & 1.228 & 1.273 & 1.615 &  &  1.372 & 1.175 & 1.824 & 1.462 & 1.450 & 1.681 \\ 
			&\multicolumn{13}{c}{$\sigma$ = 0.2}\\
			10$\backslash$35 & 1.334 & 1.510 & 1.838 & 1.494 & 1.817 & 1.797 & & 1.291 & 1.484 & 1.547 & 1.496 & 1.518 & 1.753 \\ 
			20$\backslash$35 & 1.172 & 1.396 & 1.791 & 1.386 & 1.821 & 1.727 & &1.120 & 1.359 & 1.475 & 1.481 & 1.531 & 1.781 \\ 
			10$\backslash$100 & 0.986 & 1.143 & 1.660 & 1.023 & 1.297 & 1.634 &  & 1.044 & 1.090 & 1.470 & 1.298 & 1.393 & 1.601\\ 
			20$\backslash$100 & 0.916 & 1.063 & 1.688 & 1.001 & 1.171 & 1.483 & &  1.078 & 1.044 & 1.533 & 1.211 & 1.267 & 1.712 \\ 
			&\multicolumn{13}{c}{$\sigma$ = 0.4}\\
			10$\backslash$35 & 1.183 & 1.528 & 1.819 & 1.218 & 1.531 & 1.642 & & 1.054 & 1.429 & 1.482 & 1.388 & 1.654 & 1.643 \\ 
			20$\backslash$35 & 1.060 & 1.362 & 1.573 & 1.248 & 1.718 & 1.643 & & 0.908 & 1.315 & 1.379 & 1.357 & 1.694 & 1.621 \\ 
			10$\backslash$100 & 0.839 & 1.104 & 1.594 & 0.822 & 1.205 & 1.537 & & 0.890 & 1.217 & 1.356 & 0.945 & 1.351 & 1.514  \\ 
			20$\backslash$100 & 0.799 & 1.099 & 1.632 & 0.849 & 1.271 & 1.645 &  & 0.790 & 1.041 & 1.333 & 0.962 & 1.350 & 1.579 \\ 
			&\multicolumn{13}{c}{$\sigma$ = 0.6}\\
			10$\backslash$35 & 0.982 & 1.485 & 1.746 & 1.164 & 1.633 & 1.714 &  & 0.912 & 1.424 & 1.643 & 1.205 & 1.599 & 1.721 \\ 
			20$\backslash$35 & 0.940 & 1.404 & 1.604 & 1.108 & 1.659 & 1.580 & & 0.819 & 1.377 & 1.380 & 1.152 & 1.587 & 1.533 \\ 
			10$\backslash$100 & 0.758 & 1.130 & 1.429 & 0.856 & 1.166 & 1.437 &  & 0.691 & 1.123 & 1.249 & 0.904 & 1.205 & 1.382 \\ 
			20$\backslash$100 & 0.724 & 1.030 & 1.471 & 0.789 & 1.208 & 1.582 &  & 0.556 & 1.095 & 1.212 & 0.827 & 1.247 & 1.538 \\ 
			&\multicolumn{13}{c}{$\sigma$ = 0.8}\\
			10$\backslash$35 & 0.853 & 1.388 & 1.560 & 1.123 & 1.627 & 1.763 &  & 0.671 & 1.073 & 1.161 & 1.056 & 1.641 & 1.722 \\ 
			20$\backslash$35 & 0.881 & 1.294 & 1.609 & 1.042 & 1.651 & 1.735 & & 0.639 & 1.092 & 1.157 & 1.094 & 1.644 & 1.756 \\ 
			10$\backslash$100 & 0.699 & 1.090 & 1.461 & 0.721 & 1.222 & 1.519 &  & 0.602 & 1.082 & 1.176 & 0.851 & 1.376 & 1.678 \\ 
			20$\backslash$100 & 0.639 & 1.076 & 1.509 & 0.692 & 1.232 & 1.593 & & 0.545 & 1.050 & 1.157 & 0.807 & 1.387 & 1.685  \\ 
			\hline
		\end{tabular}}
		\caption{Monte Carlo simulation results. The table reports the out-of-sample mean squared errors (MSE) of three models: individual time-series (TS), aggregate-on-components (AC), and aggregate-on-aggregate (A), all relative to the pooled model (P). Values greater than 1 indicate that the pooled model outperforms the respective method. The data generating process includes heterogeneous intercepts and slopes, with $\sigma$ controlling the level of heterogeneity. Results are shown for both Gaussian and student-$t(5)$ innovations, as well as for two different numbers of predictors ($p=50$ and $p=500$). Each entry is based on 1000 Monte Carlo simulations. \label{tab:simul}}
	\end{table}
\end{center}

\section{Nowcasting EU output}\label{sec:eunowcasting}

In the section we describe the machine learning MIDAS panel data models that we use for nowcasting. 
Then, we present the aggregation schemes used to obtain EA-19 GDP growth nowcasts. Finally, we discuss the data and the empirical results.

\subsection{ML MIDAS panel data models\label{sec:method}}

Our approach builds on \cite{babii2022machine}, who introduced machine learning MIDAS (ML MIDAS) regressions with an application to single series nowcasting, and \cite{babii2022machinepanel}, who extended the method to panel data settings. Moreover, \cite{babii2022panelpe} consider an application of nowcasting price-earnings ratios for a large set of US firms. In Online Appendix Section \ref{sec:method} we provide a detailed discussion of these types of models.

\medskip

These models fit into a generic linear panel regression setting, albeit with regularization to take account of the high-dimensional nature of the regressors. We follow the notation in the aforementioned papers.
Define $\mathbf{y}_i$ = $(y_{i,1+h},\dots,y_{i,T+h})^\top$, with $h$ the forecasting/nowcasting horizon, $\mathbf{\tilde y}_{i,q}$ = $(y_{i,1-q},\dots,y_{i,T-q})^\top$ for $q\in[Q]$ lagged dependent variables, $\mathbf{\tilde y}_i = (\mathbf{\tilde y}_{i,1}, \dots, \mathbf{\tilde y}_{i,Q})$, and $\mathbf{u}_i$ = $(u_{i,1},\dots,u_{i,T})^\top.$
Stacking time series observations in vectors, the regression equation for each $i\in[N]$ with pooling of covariates $\mathbf{x}_i$ is:
\begin{equation*}
	\mathbf{y}_i = \alpha \iota_T + \mathbf{\tilde y}_i\rho_i + \mathbf{x}_i\beta + \mathbf{u}_i,
\end{equation*}
where $\iota_T$ is a size $T$ vector of ones, $\rho_i\in\R^{Q}$ coefficients of autoregressive lags, and $\beta\in\R^{LK}$ is a vector of regression slopes. We also define the vector of all outcomes $\mathbf{y} = (\mathbf{y}_1^\top,\dots, \mathbf{y}_N^\top)^\top$, regressors $\mathbf{X}=(\mathbf{x}_1^\top, \dots, \mathbf{x}_N^\top)^\top$, and errors $\mathbf{u} = (\mathbf{u}_1^\top,\dots,\mathbf{u}_N^\top)^\top$. Then stacking all observations together, we obtain
\begin{equation*}
	\mathbf{y} = \alpha \iota_{NT} + \mathbf{\tilde Y}\rho + \mathbf{X}\beta + \mathbf{u},
\end{equation*}
where $\mathbf{\tilde Y}$ is a diagonal matrix with elements $\mathbf{\tilde y}_i,i=1,\dots,N$, and $\rho=(\rho_1^\top,\dots,\rho_N^\top)^\top.$
To deal with the large number of predictors, we use the sg-LASSO regularization that was used successfully for individual time series machine learning regressions in \cite{babii2022machine}. The MIDAS approach reduces efficiently the dimensionality of high-frequency lag coefficients. An alternative approach, known as the U-MIDAS, see \cite{foroni2015unrestricted}, would estimate all individual coefficients associated with high-frequency covariate lags hoping that machine learning would pick up relevant lags. This strategy is not appealing because we always pay a price for the model selection which can be substantial with heavy-tailed time series data, typically leading to worse predictions compared to regularized MIDAS schemes; see \cite{babii2022machine, babii2022machinepanel} for further discussion and details. The pooled panel data estimator with heterogeneous autoregressive dynamic and a variation of sparse-group LASSO solves
\begin{equation}\label{eq:sgl}
	\min_{(a,b,c)\in\R^{1+NQ+LK}}\|\mathbf{y} - a \iota_{NT} - \mathbf{\tilde Y}b  - \mathbf{X}c\|_{NT}^2 + 2\lambda\Omega_\gamma(b,c), 
\end{equation}
where $\|.\|_{NT}^2 = |.|^2/(NT)$ is the scaled $\ell_2$ norm squared and
$\Omega_\gamma(b,c)$ = $\gamma\left[|b|_1 + |c|_1\right] + (1-\gamma)\left[\|b\|_{2,1} + \|c\|_{2,1}\right]$
is a penalty, which is a linear combination of LASSO and group LASSO penalties. The weight parameter $\gamma\in[0,1]$ determines the relative importance of the $\ell_1$ (sparsity) and the $\ell_{2,1}$ (group sparsity) norms. The amount of regularization is controlled by the regularization parameter $\lambda\geq 0$. Recall also that, for a group structure $\mathcal{G}$ described as a partition of $[p]=\{1,2,\dots,p\}$, the group LASSO norm is $\|u\|_{2,1}=\sum_{G\in\mathcal{G}}|u_G|_2$, where $u$ is a generic vector and $u_G$ are the elements of $u$ corresponding to group $G$. We assume that the group structure is observed, which in our setting corresponds to: a) country-specific autoregressive lags, and b) time series lags of covariates. It is also feasible to combine covariates of a similar nature into groups. 

\medskip

In addition to the pooled panel models, as in Section \ref{sec:theory}, we also look at variations involving individual country models. In the discussion of the empirical results, we will refer to the following models:
	\begin{itemize}
		\item \textit{Time series}: time series regressions on Euro area aggregates;
		\item \textit{Pooled}: ML MIDAS panel model with pooled auto-regressive coefficients;
		\item \textit{HetAR}: ML MIDAS panel model with country-specific autoregressive coefficients;
		\item \textit{Country-specific}: ML MIDAS model without panel component.
	\end{itemize}

\subsection{Aggregation schemes}

We use the following schemes to combine country-level GDP growth nowcasts  $y_{i,t|\tau}$ for country $i$ in quarter $t$, given the high-frequency information available up to $\tau$, and compute the aggregate Euro area level GDP nowcasts, denoted $y_{\text{ea},t|\tau}$:

\begin{enumerate}
	\item We compute the weights based on historical proportions of absolute GDP growth. Each country weight at quarter $t$ is
		$W^{(1)}_{i|t}$ = $\frac{\sum_{j=1}^{t-1} |y_{ij} |}{\sum_{i=1}^{N}\sum_{j=1}^{t-1} |y_{ij}|}.$
	\item We compute the weights based on the most recent proportions of absolute GDP growth. Each country weight at quarter $t$ is
		$W^{(2)}_{i|t}$ = $\frac{|y_{it-1}|}{\sum_{i=1}^{N} |y_{it-1}|}.$
	\item We compute the weights based on the most recent proportions of GDP level. Each country weight at quarter $t$ is
		$W^{(3)}_{i|t}$ = $\frac{\text{GDP}_{it-1}}{\sum_{i=1}^{N} \text{GDP}_{it-1}}.$
	\item We project the weights onto the most recent Euro area GDP growth series $y_{ea,t-1}$. Each country weight at quarter $t$ is obtained from regression
		$y_{ea,t-1}$ = $\sum_{i=1}^{N} W^{(4)}_{i|t} y_{it-1} + \epsilon_{t-1},$
	with $W^{(4)}_{i|t}$ constrained to be positive for each $i$ and their sum across countries equals one. We estimate $W^{(4)}_{i|t}$ using constrained least squares. We note that to ensure stability of the optimization, we add a small ridge penalty of $10^{-7}$. 
\end{enumerate}
For all weighting schemes we use $t-1$ GDP data, since at quarter $t$ this is the information available in real-time. The Euro area aggregate nowcast is computed as $ y_{\text{ea},t|\tau} = \sum_{i=1}^N W^{(q)}_{i|t}   y_{i,t|\tau},$ for $q\in\{1,2,3,4\}$, corresponding to the set of weights.

\medskip

Finally, we also consider what we call the Euro area model. In this model, we nowcast Euro area GDP growth based on aggregate Euro area data. For this, we use the machine learning MIDAS setup of \cite{babii2022machine}.

\subsection{Standard macro releases \label{sec:standardmacro}}
We use real-time vintages of standard macro monthly releases from several sources (see the Appendix for the details of each series). GDP growth is quarterly and is available in real-time in our sample. The first vintage is January 2015, for which we have GDP vintages for all EA-19 countries. 
As predictor variables, we collect 64 traditional macroeconomic series: 57 are monthly series, three are weekly, and four are daily. 
Monthly series are hard data such as the unemployment rate and industrial production; 47 are available at the country level while the remaining 10 are Euro area aggregates; weekly series are oil products which are available at the country level; four daily series are financial markets data covering stock market, gold, foreign exchange, and interest rates.

\medskip
It is worth mentioning that some macro series are available with a month, two, or even three months of delay relative to the nowcasting month. Since we use real-time data vintages, we naturally take into account such delays. This is particularly important when we analyze the additional gain in nowcasting accuracy when using news data---which we describe below---since such data are timely and available without delays.

\subsection{News data}
We collect news data from Dow Jones Factiva.
The data set contains daily printed and online full-text articles from three different sources dealing with economic and financial issues, namely The Economist, Reuters News, and The Wall Street Journal.
The final data set consists of approximately $2.5$ million articles and 1 billion words from January 2005 to December 2022.
We construct news-based indicators following the fine-grained aspect-based sentiment (FiGAS) by \cite{consoli2022fine}.
This is a rule-based algorithm that provides sentence-level sentiment scores  for textual information in the English language. 
The sentiment is (i) \textit{aspect-based}, meaning that it analyzes only the words in a sentence that relate to a specific topic of interest  based on linguistic dependencies, and (ii) \textit{fine-grained}, that is, the sentence-level sentiment score comes from a human-annotated dictionary tailored for economic and financial applications and is defined in $[-1,+1]$.

\medskip

We compute news-based indicators for the following three topics covering different aspects of economic and financial activity: \textit{economy}, \textit{employment}, and \textit{inflation}.
Each topic is associated with a set of additional keywords that we look for in the articles: for instance, the \textit{economy} topic also includes related terms, like GDP, output, or economic growth.  
We compute country-specific indicators for all EA-19 member states by filtering only on sentences where there is a direct mention of the country in the analysis. 
We refer to the online appendix for additional details about the news-based sentiment indicators, the selection of the topics' keywords, and the construction of country-specific measures.

\medskip

The output of FiGAS consists of daily news-based indicators of \textit{sentiment} and \textit{volume} for each topic and country. 
The sentiment measures are obtained by summing the sentence-level sentiment scores for all articles within that day, while the volume corresponds to the number of sentences containing an explicit mention of the topic. 
Most importantly, the news-based sentiment indicators are real-time and we include them as additional regressors with no publication delay.
In the remainder, our baseline model will include news-based indicators about the \textit{economy}, \textit{employment}, and \textit{inflation}. 
\subsection{Empirical results \label{sec:results}}

We apply the machine learning regressions described in Section \ref{sec:method} to assess whether the aggregate Euro area GDP growth nowcasts are more accurate than panel data models which are based on individual country-level data and several weighting schemes of individual country nowcasts to construct an aggregate nowcast.

\paragraph{Aggregate versus panel data regressions:}
Table \ref{tab:empirics_main_aggregate} reports the nowcast comparison between the (aggregate) Euro area model and the panel data models, namely the pooled and the HetAR panel models, and the country-specific MIDAS regressions.  
The forecast accuracy is measured as root-mean-squared errors (RMSEs) at three nowcast horizons (i.e., 2- and 1-month ahead and end of the quarter).
The first row in Panel A of Table \ref{tab:empirics_main_aggregate} reports the RMSEs of the aggregate Euro area benchmark model, while the other rows report the RMSEs of the proposed panel data models relative to the benchmark: values below unity signal a better performance of the proposed model relative to the benchmark. 
For each panel data model, we document the performance of the four different weighting schemes ($W^{(1)}$-$W^{(4)}$) to aggregate individual country forecasts as described in Section \ref{sec:eunowcasting}.
Columns 1 to 3 report the results on the full sample, while columns 4 to 6 consider only the pre-COVID period.

\begin{table}[]
	\centering
	\begin{tabular}{rccc c ccc}
		\hline
		& 2-month & 1-month & EoQ & & 2-month & 1-month & EoQ \\
		
		& \multicolumn{3}{c}{\it Full sample} &&\multicolumn{3}{c}{\it Pre-COVID}\\[0.2em]
		& \multicolumn{6}{c}{Panel A. \it Time series}\\
		EA-19 & 4.960 & 3.642 & 3.366 && 0.185 & 0.215 & 0.176 \\
		+ ADS & 0.830 & 1.043 & 0.886$^{**}$ && 0.961 & 1.146 & 1.031 \\
		+ other countries & 1.501 & 2.136 & 1.685 && 1.653 & 1.050 & 0.753 \\
		+ countries \& ADS & 1.501 & 2.136 & 1.536 && 1.653 & 1.039 & 0.753 \\
		& \multicolumn{6}{c}{Panel B. \it Pooled}\\
		$W^{(1)}$ & 0.693 & 0.814 & 0.750$^{**}$ &&  2.420 & 1.880 & 2.551 \\
		$W^{(2)}$ & 0.681 & 0.775 & 0.704$^{**}$ &&  3.201 & 2.154 & 2.973 \\
		$W^{(3)}$ & 0.680 & 0.729 & 0.664$^{**}$ &&  1.802 & 1.880 & 2.171 \\
		$W^{(4)}$ & 0.661 & 0.722 & 0.663$^{**}$ &&  1.640 & 1.883 & 2.168 \\
		& \multicolumn{6}{c}{Panel C. \it HetAR}\\
		$W^{(1)}$ & 0.677 & 0.817 & 0.756$^{**}$ &&  2.393 & 1.640 & 2.268 \\
		$W^{(2)}$ & 0.679 & 0.789 & 0.729$^{**}$ &&  3.664 & 2.074 & 2.893 \\
		$W^{(3)}$ & 0.669 & 0.738 & 0.663$^{**}$ &&  1.792 & 1.588 & 1.858 \\
		$W^{(4)}$ & 0.684 & 0.743 & 0.665$^{**}$ &&  1.684 & 1.584 & 1.835 \\
  	& \multicolumn{6}{c}{Panel D. \it Country-specific regressions}\\
		$W^{(1)}$ & 0.743 & 0.984 & 0.933$^{**}$ && 3.229 & 1.576 & 1.800 \\
		$W^{(2)}$ & 0.803 & 1.063 & 0.932$^{**}$ && 5.981 & 2.037 & 2.277 \\
		$W^{(3)}$ & 0.777 & 0.996 & 0.870$^{**}$ && 1.035 & 0.557 & 0.755 \\
		$W^{(4)}$ & 0.799 & 1.046 & 0.835$^{**}$ && 0.882 & 0.548 & 0.743$^{**}$ \\
		\hline
	\end{tabular}
	\medskip
	\caption{\small Nowcast comparison table --- root-mean-squared errors (RMSEs). 
		Nowcast horizons are 2- and 1-month ahead and end of the quarter (EoQ). 
		The first row reports the absolute RMSE for the EA-19 model using only aggregate data. All remaining rows are relative RMSEs vis-\`a-vis the first row. 
        Panel A displays results for time series regressions in which we directly nowcast Euro area GDP growth. 
        Panels B and C report results for panel regression models and different weighting schemes.
		Rows $W^{(1)}$-$W^{(4)}$ denote different weights used to compute the Euro area aggregate based on individual country nowcasts.
		For rows $W^{(1)}$-$W^{(4)}$, we report RMSEs relative to the Euro area model:   values smaller than unity indicate an improvement in prediction with respect to the benchmark.
		Panel B reports results for the \textit{Pooled} panel model for 19 countries, while Panel C for the \textit{HetAR} model. 
        Panel D reports results using country-specific MIDAS regression nowcasts aggregated based on weighting schemes $W^{(1)}$-$W^{(4)}$.
		We use 5-fold cross-validation adjusted for panel data to compute both sg-LASSO tuning parameters, see the Online Appendix Section \ref{cv.details} for further details.
        Out-of-sample periods: 2016 Q1 - 2019 Q4 (pre-COVID) and 2016 Q1 - 2022 Q4 (Full sample).
       Significance: ** $p < 0.05$; * $p < 0.1$ with the equally-weighted multi-horizon superior predictive ability test of \cite{quaedvlieg2021multi}: asterisks next to EoQ indicate a significantly better performance with respect to the benchmark for the Full or Pre-COVID samples.
		\label{tab:empirics_main_aggregate}} 
\end{table}

\medskip
Panel A reports the results for the time series regressions when we nowcast directly Euro area aggregates. 
While the first row regressions include only aggregated information about the Euro area as an explanatory variable (option (a) in the Introduction), the other rows in the top panel add respectively the \cite{aruoba2009real} (ADS) index and information from individual EA-19 countries as additional regressors (option (b) in the Introduction).
A few patterns are worth highlighting.
First, the model's accuracy is largely impacted by the inclusion of the COVID period observations, with the full sample results showing much larger errors than the pre-COVID sample.
Second, the inclusion of timely information about the US business cycle proxied by the ADS index does not bring any systematic improvement with respect to the benchmark model.
Third, an even worse performance results from the inclusion of information about individual EA-19 countries into the time series regression model compared to the benchmark model only using aggregate EA data. Hence, option (a) is better than (b), and therefore using only the aggregate data in machine learning models suffices.

\medskip
Turning to the performance of the panel data machine learning models in Panels B and C, we obtain mixed results. For the pre-COVID sample, we note that the  panel data models under perform vis-\`a-vis the benchmark.
In contrast, the panel data models always outperform the benchmark when looking at the full sample. 
This result is robust with respect to the choice of the weighting scheme and the nowcasting horizon.
The proposed models exploit the additional country-specific information included in the panel data structure: this information turns out to be redundant during normal times, while it proves relevant during the COVID pandemic. 
Comparing the two panel data models, the HetAR model generally achieves better performances than the pooled models: the inclusion of heterogeneity in the form of country-specific lags seems to be a better choice than pooling all coefficients.

\medskip

Panel D of Table \ref{tab:empirics_main_aggregate}  reports the results with country-specific MIDAS regressions -- hence not exploiting the panel data structure, but instead estimating single regressions per country -- which are then aggregated using again the same weighting schemes. Hence, Panel D differs from Panels B and C only with respect to the panel structure, while the information set, the MIDAS structure and the weighting schemes are the same.
Country-specific regressions aggregated following weighting schemes $W^{(3)}$ and $W^{(4)}$ attain good results both when looking at the full sample and in the pre-COVID period. 
On the one hand, Panel D shows gains with respect to the Euro area benchmark during the pre-COVID sample, when panel models perform poorly. 
The gains attained by country-specific regressions are large and up to 45 percentage points. 
On the other hand, the performance in the full sample, although still better than the benchmark, does not improve with respect to panel models which perform best when including the COVID pandemic in the analysis.


\paragraph{The Big Four:} We test whether nowcasting the four largest Euro area countries---France, Germany, Italy, and Spain (i.e., the ``Big Four")---separately may give an advantage in producing more accurate predictions for the Euro area GDP. See also \cite{ashwin2021nowcasting, barbaglia2022testing, cascaldigarcia2021}, among others.
We compare the performance of the HetAR and pooled panel models, as well as the country-specific regressions as in the previous section. To compute the aggregated Euro area nowcasts we use weighting scheme $W^{(4)}$ modified to include only the Big Four.

\medskip

Table \ref{tab:bigfour_smallcountries} reports the results of the four-country models relative to the benchmark model appearing in the first row of Table \ref{tab:empirics_main_aggregate}. All point to a better performance of the model involving nowcasting only the largest European economies.  
This result is robust across horizons, model specifications and, interestingly, the full and pre-COVID samples. Indeed, panel data models with the Big Four attain a more accurate performance even when considering the pre-COVID sample, whereas that was not the case in Table \ref{tab:empirics_main_aggregate}. 
However, looking at the full sample, panel data models with all Euro area countries always outperform the Big Four models.
Hence, rather than focusing only on the indicators from the largest European countries, our evidence suggests that there is a value-added in the nowcasting model for all EA-19 countries when looking at the full sample and including the COVID in our model.
Looking at the pre-COVID sample, the Big Four models are outperformed at all horizons by the country-specific regressions using all Euro area countries reported in Table \ref{tab:empirics_main_aggregate}.
Overall, this suggests the best results are obtained with the inclusion of information from all EA-19 countries both in the full sample (with panel models) and pre-COVID (with country-specific regressions).


\begin{table}[]
	\centering
	\begin{tabular}{rccc c ccc}
		\hline
		& 2-month & 1-month & EoQ & & 2-month & 1-month & EoQ \\
		
		& \multicolumn{3}{c}{\it Full sample} &&\multicolumn{3}{c}{\it Pre-COVID}\\[0.2em]
		& \multicolumn{6}{c}{Panel A. \it Big four --- Pooled}\\
		$W^{(4)}$& 0.698 & 0.871 & 0.808$^{**}$ && 0.877 & 0.978 & 0.820 \\
		& \multicolumn{6}{c}{Panel B. \it Big four --- HetAR}\\
		$W^{(4)}$ & 0.692 & 0.743 & 0.734$^{**}$ && 0.810 & 1.025 & 0.884 \\
		& \multicolumn{6}{c}{Panel C. \it Big four --- country-specific regressions}\\
		$W^{(4)}$ & 0.813 & 1.046 & 0.820$^{**}$ && 0.965 & 0.557 & 0.807$^{*}$ \\
		\hline
	\end{tabular}
	\medskip
	\caption{\small Nowcast comparison table --- root-mean-squared errors (RMSEs). 
			Nowcast horizons are 2- and 1-month ahead and end of the quarter (EoQ). 
			All entries are relative RMSEs to the first row of Table \ref{tab:empirics_main_aggregate}. Panel A and B display results for Pooled and HetAR panel regression models, respectively, based on France, Germany, Italy, and Spain (i.e.\ the ``Big Four") countries. 
            Panel C displays results for the Big Four country-specific regressions aggregated based on weighting scheme $W^{(4)}$. 
            We use 5-fold cross-validation adjusted for panel data to compute both sg-LASSO tuning parameters, see the Online Appendix Section \ref{cv.details} for further details.
            Out-of-sample periods: 2016 Q1 - 2019 Q4 (pre-COVID) and 2016 Q1 - 2022 Q4 (Full sample). Significance: ** $p < 0.05$; * $p < 0.1$ with the equally-weighted multi-horizon superior predictive ability test of \cite{quaedvlieg2021multi}: asterisks next to EoQ indicate a significantly better performance with respect to the benchmark for the Full or Pre-COVID samples.
		\label{tab:bigfour_smallcountries}} 
\end{table}

\medskip
To explain the relative improvement in nowcasting performance achieved by using the full panel of Euro area countries as opposed to just the Big Four, several factors come into play. First, from a modeling perspective, machine learning panel regression models are estimated more accurately when $N$, i.e.\ the number of countries, is large - see \cite{babii2022machinepanel}. 
More accurate parameter estimates lead to higher quality nowcasts. 
Second, the real-time flow of data and information varies among countries. For instance, Belgium, a nation falling outside the Big Four category, consistently boasts superior and more timely survey data, significantly contributing to the accuracy of Euro area GDP predictions, see  \cite{basselier2018nowcasting}. 
From an economic perspective, given Belgium's strong economic ties with Germany and France due to its geographical proximity, its economic news serves as a reliable signal for the broader economies, therefore influencing the overall Euro area GDP projection. 
By incorporating the entire panel of countries, we can effectively capture and harness these effects. 
The Online  Appendix provides additional insights about the drivers of the nowcasting performance and a country-level evaluation.

\paragraph{Nowcast aggregation and combination:} We now focus on how to aggregate and combine individual nowcasts, starting from the nowcasting performance of the four proposed weighting schemes. 
Looking at the full sample results of Table \ref{tab:empirics_main_aggregate}, we observe that $W^{(4)}$ achieves the lowest RMSEs for the pooled panel model. 
This result does not hold for the HetAR model, where the $W^{(3)}$ weighting scheme produces slightly more accurate forecasts than $W^{(4)}$.
Regarding country-specific regressions, $W^{(3)}$ and $W^{(4)}$ attain similar performance in both the full sample and the pre-COVID period, and substantially outperform the other weighting schemes.
Figure \ref{fig:weights} illustrates the four estimated weighting schemes, with $W^{(3)}$ and $W^{(4)}$ notably exhibiting denser characteristics than the other schemes. 
Consequently, nowcasts employing denser weighting schemes yield more accurate results, further reinforcing the argument in favor of utilizing the entire panel of European countries to improve nowcasting precision. In the Online Appendix Section \ref{appsec:weights} we further investigate the nowcasting performance across weighting schemes, where we compare the weights obtained by $W^{(4)}$ (i.e., projections on GDP) and by $W^{(3)}$ (i.e., proportion of GDP level).
Compared to $W^{(3)}$, we observe that $W^{(4)}$ assigns smaller weights to Germany and to a lesser extent the Netherlands, while it inflates the relative importance of some small- and medium-sized countries, namely Austria, Belgium and Luxembourg.
Although the size of these economies within the Euro area is relatively small, information about economic developments in those countries plays a relevant role in attaining more accurate nowcasts. 

\begin{figure}[]
	\centering
	\makebox[\textwidth][c]{\includegraphics[scale=0.4]{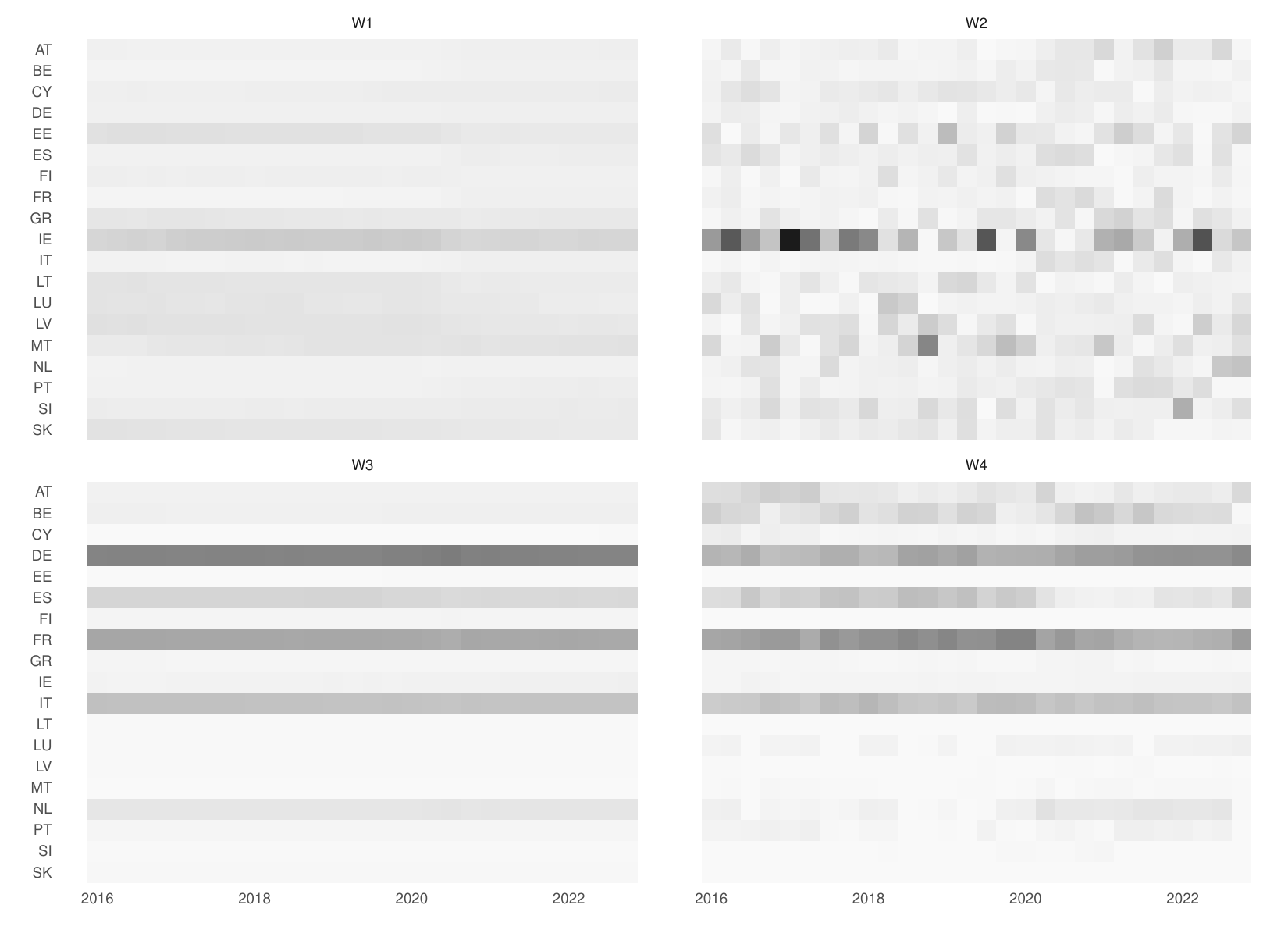}}%
	\caption{Estimated weights for the weighting schemes $W^{(1)}$-$W^{(4)}$ by country. 
    Each tile corresponds to a quarter. 
    The darker the color, the larger the weight.}
	\label{fig:weights}
\end{figure}

Finally, we explore the relative importance of smaller countries against the Big Four by plotting in Figure \ref{fig:weights_small15big4} the sum of weights of the 15 small Euro area countries (\textit{Small15}) relative to the sum of the weights of the Big Four for weighting schemes $W^{(3)}$ (dashed line) and $W^{(4)}$ (solid line).
While the relative importance of smaller countries with respect to the Big Four is stable when looking at the $W^{(3)}$ weights (indeed, proportions of GDP level are steady in time and vary only at a slow pace), we observe large variability when considering $W^{(4)}$. 
Smaller countries are relatively more important in 2016-17 and, most notably, after 2020, when the total weight assigned to smaller countries doubles, going from 0.3 to approximately 0.6. 
The information coming from smaller Euro area countries plays an important role in attaining an accurate nowcast in the years following the COVID-19 pandemic, hence explaining the good performance of panel models on all EA-19 countries in the full sample reported in Table \ref{tab:empirics_main_aggregate}.

\begin{figure}[]
	\centering
	\makebox[\textwidth][c]{\includegraphics[scale=0.55]{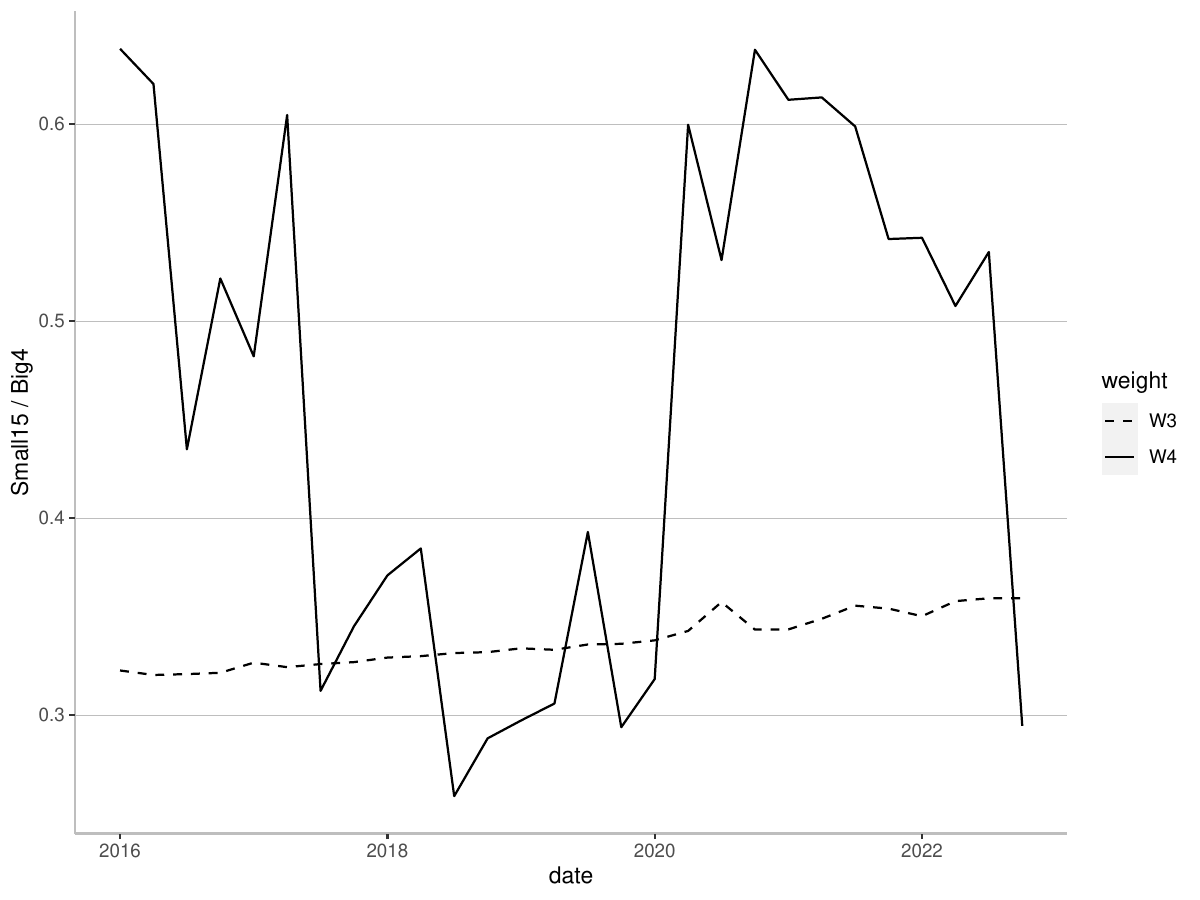}}%
	\caption{Sum of weights of the 15 small Euro area countries (\textit{Small15}) relative to the sum of the weights of the Big Four (\textit{Big4}) for weighting schemes $W^{(3)}$ and $W^{(4)}$.}
	\label{fig:weights_small15big4}
\end{figure}

Overall, our results highlight a heterogeneous nowcast performance of the analyzed models: panel models perform best when considering the full sample and information from all Euro area countries, country-specific MIDAS regressions outperform all other models in the pre-COVID sample, while nowcasting models on the Big Four attain a good performance across all time samples, although not the best one.
Given this heterogeneity in performance, we combine forecasts obtained from individual models following the linear combination method by \cite{stock2004}.
In particular, we combine the forecasts obtained from the following models: the \textit{Pooled} panel model and \textit{Country-specific regressions} on all EA-19 countries (Panels B and D, respectively, of Table \ref{tab:empirics_main_aggregate}), and the \textit{HetAR} model on the Big Four (Panel B of Table \ref{tab:bigfour_smallcountries}). 
For all individual forecasts, we consider the aggregated forecasts obtained with $W^{(4)}$ weights.
We select forecasts aggregated with $W^{(4)}$ weights since they provide the best nowcasting performance overall.
Regarding the selection of the individual nowcasting models, we include all model types analyzed in the paper (\textit{Pooled} and \textit{HetAR} panels, as well as country-specific MIDAS regressions) taking into account the time period and information set (i.e., all EA-19 or only the Big Four countries) where they achieve their best performance. 
We have also experimented with other model subsets for which we obtained similar results.

\medskip
The nowcast combination results are reported in Table \ref{tab:combination}. We start with Panel A which covers models with news data. 
Combining the nowcasts consistently delivers large gains which range between a 20 to 30\% reduction with respect to the benchmark RMSEs. 
Moreover, the gains are robust across horizons and, importantly, full versus pre-COVID samples. Indeed, combinations attain more accurate nowcasts both in the full sample and in the pre-COVID period, reaching a performance that is close to the best-performing individual model. 
Panel B of Table \ref{tab:combination} explores one final aspect of our analysis, namely the added value of the news indicators. Compared to Panel A showing the nowcasting performance of models including news indicators, Panel B excludes them. 
The inclusion of news delivers nowcasting gains, even though these gains are not large, ranging between 1-5\%. Therefore,
although only marginally, the inclusion of news indicators positively impacts the accuracy of the nowcasts.


\begin{table}[h]
	\centering
	\begin{tabular}{rccc c ccc}
		\hline
		& 2-month & 1-month & EoQ & & 2-month & 1-month & EoQ \\
		
		& \multicolumn{3}{c}{\it Full sample} &&\multicolumn{3}{c}{\it Pre-COVID}\\[0.2em]
		& \multicolumn{6}{c}{Panel A. \it Combination -- with news}\\
		$W^{(4)}$ & 0.693 & 0.815 & 0.717$^{**}$ && 0.831 & 0.770 & 0.733$^{**}$ \\
		& \multicolumn{6}{c}{Panel B. \it Combination -- without news}\\
		$W^{(4)}$  & 0.701 & 0.832 & 0.725$^{**}$ &&  0.890 & 0.801 & 0.773$^{*}$\\
		\hline
	\end{tabular}
	\medskip
	\caption{\small Nowcast comparison table --- root-mean-squared errors (RMSEs). 
		Nowcast horizons are 2- and 1-month ahead and end of the quarter (EoQ). 
		All entries are relative RMSEs to the first row of Table \ref{tab:empirics_main_aggregate}. Panel A reports results for forecast combination method combining models \textit{Pooled} (Table \ref{tab:empirics_main_aggregate} Panel B), \textit{Country-specific regressions} (Table \ref{tab:empirics_main_aggregate} Panel D), and
        \textit{Big four — HetAR} (Table \ref{tab:bigfour_smallcountries} Panel B) based on $W^{(4)}$ weights. 
        Panel B reports forecast combination results for the same models as in Panel A but excluding news data from the predictor variable set. We use 5-fold cross-validation adjusted for panel data to compute both sg-LASSO tuning parameters, see the Online Appendix Section \ref{cv.details} for further details. Out-of-sample periods: 2016 Q1 - 2019 Q4 (pre-COVID) and 2016 Q1 - 2022 Q4 (Full sample). Significance: ** $p < 0.05$; * $p < 0.1$ with the equally-weighted multi-horizon superior predictive ability test of \cite{quaedvlieg2021multi}: asterisks next to EoQ indicate a significantly better performance with respect to the benchmark for the Full or Pre-COVID samples.
        \label{tab:combination}} 
\end{table}

\section{Conclusion}
\label{sec:conclusion}

The paper studies the Euro area nowcasting using MIDAS machine learning panel data regression models. These models offer the flexibility needed to analyze extensive datasets gathered from diverse sources, sampled at varying frequencies, and available in both cross-sectional and time series dimensions. Through this research, we introduce several innovative empirical findings that carry significant relevance for policymakers. Our findings reveal a significant enhancement in the accuracy of nowcasts for the Euro area aggregate GDP when incorporating data from smaller countries. These improvements are substantial, reaching up to 30\%, and remain robust across nowcasting horizons. We attribute these gains to two primary factors.

\medskip

Firstly, in highly turbulent times such as the COVID-19 pandemic, the parameter estimates of the machine learning panel data regression models are estimated more precisely when a broader cross-sectional dataset is employed. 
Secondly, smaller countries, such as Austria, Belgium, or Luxembourg, in contrast to larger Euro-area countries like Germany or France, tend to possess higher-quality and more timely survey data. 
Given the high level of economic interconnectivity, especially among neighboring countries in the Euro area, the inclusion of data from these smaller nations results in more accurate signals, ultimately leading to enhanced model predictions.

\medskip

In addition to our primary findings, we also contribute to the expanding body of literature concerning the application of alternative data, such as information extracted from newspaper articles, to enhance the accuracy of economic forecasts. The incorporation of such data is particularly advantageous due to their real-time availability and the lack of any publication delay for these data. Our research demonstrates that news data improve the accuracy of Euro area GDP nowcasts. 

\if0\blind
{
	\section*{Acknowledgments}
	We received helpful comments from Peter Reusens, Paolo Paruolo, Wouter Van der Veken, and Raf Wouters, as well as participants at a National Bank of Belgium seminar, the Nowcasting workshop at the Paris School of Economics, and the 
	5$^{th}$ Conference on ``Nontraditional Data, Machine Learning, and Natural Language Processing in Macroeconomics." 
	The views expressed are purely those of the authors and should not, in any circumstances, be regarded as stating an official position of the European Commission. 
	Eric Ghysels and Jonas Striaukas gratefully acknowledge the financial support of the National Bank of Belgium. 
	Jonas Striaukas also acknowledges the financial support of the F.R.S.\---FNRS PDR under project Nr. PDR T.0044.22.
} \fi

\if1\blind
{
	
} \fi

\clearpage
	
\bibliographystyle{econometrica}
\bibliography{midas_ml}

\clearpage

\setcounter{page}{1}
\setcounter{section}{0}
\setcounter{equation}{0}
\setcounter{table}{0}
\setcounter{figure}{0}
\renewcommand{\thesection}{S.\arabic{section}}
\renewcommand{\theequation}{S.\arabic{equation}}
\renewcommand\thetable{S.\arabic{table}}
\renewcommand\thefigure{S.\arabic{figure}}
\renewcommand\thepage{S. - \arabic{page}}
\renewcommand\thealgorithm{S.\arabic{algorithm}}

\bigskip
\begin{center}
{\large\bf SUPPLEMENTARY MATERIAL}
\end{center}

\newpage
\section{Data}

\begin{center}
	\begin{table}[htp]
		{\tiny
			\makebox[15cm]{
				\centering
				\begin{tabular}{rlllll}
					\hline
					& Description & Units & Freq. & Source & Individ. \\ 
					\hline
					1 & Gross Domestic Product & SWDA, Mil.Ch.2015.EUR & Q & Eurostat & $\checkmark$ \\ 
					2 & Long Term Government Bond Yield & Avg. (\%) & M & European Central Bank & $\checkmark$ \\ 
					3 & Harmonized Consumer Price Index & SA, 2015=100 & M & Eurostat/Haver Analytics & $\checkmark$ \\ 
					4 & HICP: Total excl. Energy \& Unprocessed Food & SA, 2015=100 & M & Eurostat/Haver Analytics & $\checkmark$ \\ 
					5 & PPI: Industry excl. Construction& NSA, 2015=100& M & Eurostat & $\checkmark$ \\ 
					6 & Exports: Volume Indices & SWDA, 2015=100 & M & Eurostat & $\checkmark$ \\ 
					7 & Imports: Volume Indices & SWDA, 2015=100 & M & Eurostat & $\checkmark$ \\ 
					8 & Industrial Production: Industry excl. Construction & SWDA, 2015=100 & M & Eurostat & $\checkmark$ \\ 
					9 & Industrial Production: Manufacturing & SWDA, 2015=100 & M & Eurostat & $\checkmark$ \\ 
					10 & IP: Electricity, Gas, Steam \& Air Cond & SWDA, 2015=100 & M & Eurostat & $\checkmark$ \\ 
					11 & Unemployment Rate & SA, \% & M & Eurostat & $\checkmark$ \\ 
					12 & Industrial Turnover: Manufacturing & SWDA, 2015=100 & M & Eurostat & $\checkmark$ \\ 
					13 & Industrial Turnover: Intermediate Goods & SWDA, 2015=100 & M & Eurostat & $\checkmark$ \\ 
					14 & Industrial Turnover: Capital Goods & SWDA, 2015=100 & M & Eurostat & $\checkmark$ \\ 
					15 & Industrial Turnover: Consumer Goods & SWDA, 2015=100 & M & Eurostat & $\checkmark$ \\ 
					16 & Retail Trade Volume excl. Autos \& Motorcycles & SWDA,2015=100 & M & Eurostat & $\checkmark$ \\ 
					17 & Balance: with EA & SWDA, Mil.EUR & M & Eurostat & $\checkmark$ \\ 
					18 & Balance: with World & SWDA, Mil.EUR & M & Eurostat & $\checkmark$ \\ 
					19 & Industrial Confidence Indicator, Percent Balance & SA, \% & M & European Commission & $\checkmark$ \\ 
					20 & Industry: Production Expectations:, Percent Balance & SA, \% & M & European Commission & $\checkmark$ \\ 
					21 & Industry: Volume of Order Books, Percent Balance & SA, \% & M & European Commission & $\checkmark$ \\ 
					22 & Industry: Stocks of Finished Products, Percent Balance & SA, \% & M & European Commission & $\checkmark$ \\ 
					23 & Industry: Volume of excl.port Order Books, Percent Balance & SA, \% & M & European Commission & $\checkmark$ \\ 
					24 & Industry: Selling Price Expectations:, Percent Balance & SA, \% & M & European Commission & $\checkmark$ \\ 
					25 & Industry: Production Trend in Recent Months, Percent Balance & SA, \% & M & European Commission & $\checkmark$ \\ 
					26 & Services Confidence Indicator & SA, \% & M & European Commission & $\checkmark$ \\ 
					27 & Consumer Confidence Indicator, Percent Balance & SA, \% & M & European Commission & $\checkmark$ \\ 
					28 & Consumer: Financial Situation last 12 Months, Percent Balance & SA, \% & M & European Commission & $\checkmark$ \\ 
					29 & Consumer: Financial Situation next 12 Months, Percent Balance & SA, \% & M & European Commission & $\checkmark$ \\ 
					30 & Consumer: Gen Economic Situation next 12 Mo, Percent Balance & SA, \% & M & European Commission & $\checkmark$ \\ 
					31 & Consumer: Major Purchases over next 12 Months, Pct Balance & SA, \% & M & European Commission & $\checkmark$ \\ 
					32 & HH Fin Situation: Sample Total, Balance & SA, \% & M & European Commission & $\checkmark$ \\ 
					33 & Consumer: Major Purchases at Present, Percent Balance & SA, \% & M & European Commission & $\checkmark$ \\ 
					34 & Consumer: Unemployment Expectations: next 12 Mo, Pct Balance & SA, \% & M & European Commission & $\checkmark$ \\ 
					35 & Construction Confidence Indicator, Percent Balance & SA, \% & M & European Commission & $\checkmark$ \\ 
					36 & Retail Trade Confidence Indicator, Percent Balance & SA, \% & M & European Commission & $\checkmark$ \\ 
					37 & Retail: Present Business Situation, Percent Balance & SA, \% & M & European Commission & $\checkmark$ \\ 
					38 & Retail: Volume of Stocks, Percent Balance & SA, \% & M & European Commission & $\checkmark$ \\ 
					39 & Retail: excl.pected Business Situation, Percent Balance & SA, \% & M & European Commission & $\checkmark$ \\ 
					40 & Retail: Orders Placed with Suppliers, Percent Balance & SA, \% & M & European Commission & $\checkmark$ \\ 
					41 & Industry: Employment Expectations:, Percent Balance & SA, \% & M & European Commission & $\checkmark$ \\ 
					42 & Constr: Factors Limiting Bldg Activity: None & SA, \% & M & European Commission & $\checkmark$ \\ 
					43 & Constr: Factors Limiting Bldg Activity: Demand & SA, \% & M & European Commission & $\checkmark$ \\ 
					44 & Constr: Limits to Bldg Activity: Weather & SA, \% & M & European Commission & $\checkmark$ \\ 
					45 & Constr:Limits to Bldg Activity:Labor Shortage & SA, \% & M & European Commission & $\checkmark$ \\ 
					46 & Constr: Limits to Bldg Activity: Eqpt Shortage & SA, \% & M & European Commission & $\checkmark$ \\ 
					47 & Constr: Limits to Bldg Activity: OthFactors & SA, \% & M & European Commission & $\checkmark$ \\ 
					48 & Euro Real Effective Exch. Rate [EER-42], CPI Deflated & NSA, Q1-99=100 & M & European Central Bank & $\checkmark$ \\ 
					49 & Trade Balance: Total & SWDA, Mil.EUR & M & Eurostat &  \\ 
					50 & Domestic PPI: Industry excl. Construction & NSA, 2015=100& M & Eurostat &  \\ 
					51 & PPI: Industry excl. Construction& NSA, 2015=100& M & Eurostat &  \\ 
					52 & Import Prices [Non EA]: Mining, Mfg \& Energy& NSA, 2015=100& M & European Central Bank &  \\ 
					53 & ECB Money Supply: M1: Annual Growth Rate & SWDA, \% & M & European Central Bank &  \\ 
					54 & Money Supply: Credit to Euro Area Residents: Loans & NSA, \%Y/Y & M & European Central Bank &  \\ 
					55 & ZEW Economic Sentiment, Current Macroeconomic Conditions & NSA, \% & M & ZEW &  \\ 
					56 & ZEW Economic Sentiment, Macro Expectations: [Next 6 Mos] & NSA, \% & M & ZEW &  \\ 
					57 & Business Climate Indicator & SA, std. dev. pts & M & European Commission &  \\ 
					58 & Euro-super 95 Pump Price excl. duties and taxes & EUR/1000L & W & European Commission & $\checkmark$ \\ 
					59 & Auto Gas Oil [Diesel] Pump Price excl. duties and taxes & EUR/1000 & W & European Commission & $\checkmark$ \\ 
					60 & Heating Gas Oil Consumer Price excl. duties and taxes& EUR/1000 & W & European Commission & $\checkmark$ \\ 
					61 & Euro STOXX 50 Index & & D & Thomson Reuters -- EIKON &  \\ 
					62 & Euro/US Dollar FX Spot Rate & & D & Thomson Reuters -- EIKON &  \\ 
					63 & Gold/Euro FX Cross Rate & & D & Thomson Reuters -- EIKON &  \\ 
					64 & Euro 3 Month Euribor & & D & Thomson Reuters -- EIKON &  \\ 
					\hline
				\end{tabular}
			} 
			\caption{\scriptsize Data description table — macro, survey, oil products and financial data. Column \textit{Description} provides a short description of each series, column \textit{Units} provides units of each series, column \textit{Freq.} is the sampling frequency, column \textit{Source} provides the original data source and \textit{Individ.} indicates whether the series is available at the country level ($\checkmark$) or not. We use Haver analytics vintage data in which the real-time data starts at 2015 January.\label{app:macro_data}} 
		}
	\end{table}
\end{center}

\begin{center}
	\begin{table}[ht]
		{\small
			\centering
			\begin{tabular}{rl}
				\hline
				Series & Countries\\ 
				\hline
				Industrial Turnover: Manufacturing SWDA, 2015=100 & Ireland  \\
				Industrial Turnover: Intermediate Goods SWDA, 2015=100 & Ireland  \\
				Industrial Turnover: Capital Goods SWDA, 2015=100 & Ireland \\
				Industrial Turnover: Consumer Goods SWDA, 2015=100 & Ireland \\
				New Motor Vehicle Registrations NSA, Units) & Cyprus \& Malta \\
				Services Confidence Indicator SA, \% Balance) & Luxembourg \\
				Retail Trade Confidence Indicator, Percent Balance SA, \% & Luxembourg \\
				Industry: Employment Expectations SA, \% & Luxembourg \\
				Constr: Factors Limiting Bldg Activity: None SA, \% & Ireland \& Malta \\
				Constr: Factors Limiting Bldg Activity: Demand SA, \% & Ireland \& Malta \\
				Constr: Limits to Bldg Activity: Weather SA, \% & Ireland \& Malta \\
				Constr: Limits to Bldg Activity: Labor Shortage SA, \% & Ireland \& Malta \\
				Constr: Limits to Bldg Activity: Equipment Shortage SA, \% & Ireland \& Malta \\
				Constr: Limits to Bldg Activity: Other Factors SA, \% & Ireland \& Malta \\
				Constr:Limits to Bldg Activity: Financial Constraints SA, \% & Ireland \& Malta \\
				\hline
			\end{tabular}
		} 
		\caption{This table reports the series that are missing for some countries. We replace those series with the Euro area data.} 
	\end{table}
\end{center}

\newpage
\section{News indicators}

We construct news-based indicators following the fine-grained aspect-based  sentiment (FiGAS) by \cite{consoli2022fine}.
This method provides sentence-level sentiment scores  for textual information in the English language. 
FiGAS has two main characteristics. 
First, it is \textit{aspect-based}, that is, it computes a sentiment score about a specific topic, rather than the overall sentiment of all sentences in a text.
Second, it is \textit{fine-grained}, meaning that words are assigned a sentiment score in [-1,+1] taken from a dictionary developed for economic and financial applications.

\medskip
If a sentence contains a direct mention of the topic of interest, then FiGAS tags each word in the sentence with its part-pf-speech and its dependency. 
Based on these tags, the algorithm checks whether the sentence corresponds to one of the linguistic rules mapped in FiGAS, which aims to capture semantic constructions that characterize the topic of interest. 
Examples of rules are the presence of an adjectival modifier or clause, or the topic of interest being followed by an object predicate.
We refer to \cite{consoli2022fine} for a detailed description of the method. 

\medskip
In our application, we compute news-based indicators for three topics covering different aspects of economic and financial activity:  \textit{economy}, \textit{employment} and \textit{inflation}.
We associate a list of additional keywords to each topic to better represent the news coverage of that economic concept as detailed in  Table \ref{tab:topics}. The selection of the keywords starts form the list of related terms extracted from the global knowledge graph from the Global Data set of Events, Language and Tone (GDELT) and it cuts off the most frequent terms based on their presence in GDELT news items. Finally, we subset the list of additional keywords by selecting only the terms that share the same polarity: for instance, we do not include in the same topic a term that has a positive connotation (e.g., employment) with a term with a negative one (e.g., unemployment). 

\begin{table}
\centering
	\begin{tabular}{ll}
  \hline
topic & keywords \\ 
  \hline
  economy & economy, gdp,  gross domestic product, economic development, \\ & economic power, economic growth, economic output, national income, revenue, \\ & investment \\ 
  employment & employment, wage, career, income, job creation, job growth, job opportunity, \\
  & labor market, labor supply, labor market, labor supply, workforce, vacancy \\ 
  inflation & inflation, inflationary, price level \\ 
\hline
\end{tabular}
\caption{List of keywords included in each topic.}
\label{tab:topics}
\end{table}

\medskip
The news-based indicators are country-specific. We rely on one additional feature of FiGAS which allows filtering sentences with a direct mention of a geographic location. 
FiGAS uses named-entity recognition to identify entities in the full-text \cite{Mohit2014}. 
If it identifies an entity, it assigns the most relevant location to the whole text, unless the sentence mentions directly another location. 
Table \ref{tab:locations} reports the list of keywords that we employ to filter the location for each Euro Area country.

\begin{table}[ht]
\centering
{\small
\begin{tabular}{ll}
  \hline
country & keywords \\ 
  \hline
AT & austria, austrian, vienna \\ 
  BE & belgium, belgian, belgica, brussels, antwerp, european capital, eu capital, wallonia, flanders \\ 
  CY & cyprus, cypriot, cypriote, nicosia \\ 
  DE & germany, german, deutschland, deutsches, berlin, frankfurt \\ 
  EE & estonia, estonian, tallinn \\ 
  GR & greece, greek, hellenic, athens \\ 
  ES & spain, espana, spanish, madrid, barcelona, catalonia, hispania, hiberia, hiberian \\ 
  FI & finland, finn, finnish, helsinki \\ 
  FR & france, french, paris, francia \\ 
  IE & eire, ire, republic of ireland, ireland, irish, dublin \\ 
  IT & italy, italian, italic, italia, rome, milan \\ 
  LT & lithuania, lithuanian, vilnius \\ 
  LU & luxembourg, luxemburger, luxembourg-ville \\ 
  LV & latvia, latvian, riga \\ 
  MT & malta, Maltese, valletta \\ 
  NL & netherlands, holland, amsterdam, rotterdam, den haag, dutch, hollander \\ 
  PT & portugal, portuguese, lisbon \\ 
  SI & slovenia, slovenian, ljubljana \\ 
  SK & slovakia, slovak, slovakian, bratislava \\
   \hline
\end{tabular}}
\caption{Euro Area 19 countries (ISO2 codes) and additional keywords for location filtering.}
\label{tab:locations}
\end{table}

\medskip
We now report two examples to show how FiGAS works. 
Consider the sentence ``the Italian economy is expected to grow more than 6\% this year after the record slump recorded in 2020" that appeared in Reuters News on December 22$^{nd}$ 2021. 
FiGAS identifies one particular linguistic rule, more specifically, this sentence consists of the topic of interest (i.e., the Italian economy) connected to a verb (i.e., is expected) followed by an open clausal complement (i.e., to grow).
Furthermore, the algorithm detects a specific mention of a location and assigns it to Italy. 
The sentence-level sentiment score is 0.6, representing the positive outlook for the Italian economy. 
Now consider this sentence ``the mighty automobile industry increased its output by 1.9\%" that was published in Reuters News on September 7$^{th}$ 2021. 
FiGAS identifies the combination of two linguistic rules that characterize the topic of interest (i.e., the automobile industry): in detail, the topic of interest is associated with two adjectival modifiers, either in the form of an adjective (i.e., mighty) or as a verb (i.e., increased).
Given that there is no direct mention of any location in this sentence, FiGAS assigns the most relevant location detected in the full article, which is Germany.
The overall sentence-level sentiment score is 0.49 and captures the positive connotation in which the text is describing the German automobile industry. 

\medskip
For each country and topic, we obtain two daily news-based indicators for \textit{sentiment} and \textit{volume}. 
The first one is obtained by summing the sentence-level sentiment scores, while the second one counts the number of sentences that discuss that topic in a day. 
As an example, Figure \ref{fig:sent_economy} plots the \textit{economy} sentiment indicators for all Euro Area 19 countries. 
The time series are aggregated monthly by summing and smoothed with a yearly rolling moving average.
There is some clear heterogeneity across the countries showing that each indicator captures country-specific economic developments.
In most cases, we notice a good fit with the business cycle fluctuations: for all major economies in the Euro Area, the \textit{economy} sentiment deteriorates clearly in correspondence with the three recessionary periods in our sample.

\begin{figure}[h]
    \centering
    \includegraphics[scale=0.55]{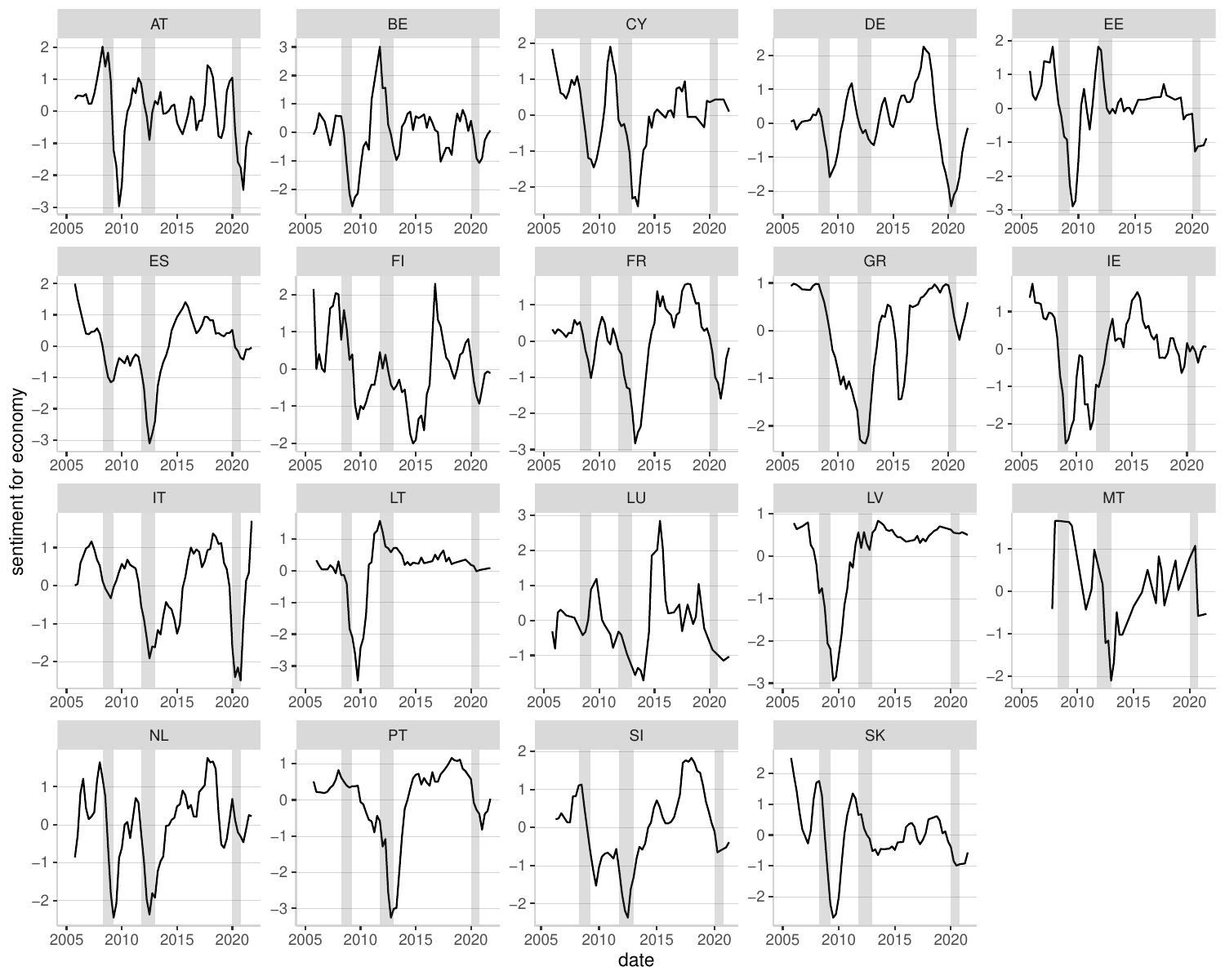}
    \caption{Sentiment indicator for \textit{economy} for the Euro Area 19 countries. The quarterly time series are smoothed with a yearly rolling moving average. The grey shadowed areas represent the recessionary periods as indicated by the Euro Area business cycle chronology.}
    \label{fig:sent_economy}
\end{figure}

\newpage

\section{Mixed-frequency panel data machine learning nowcasting models\label{sec:method}}


In Section \ref{sec:theory} of the paper we simplified the notation by ignoring the mixed frequency nature of the data. Here we introduce the MIDAS panel data models used in the empirical nowcasting study. More specifically, we work with panel data $(y_{i,t},x_{i,t})\in\mathbb{R}\times\mathbb{R}^p$, where $t=1,2,\dots,T$ denotes time and $i=1,2,\dots,N$ is the cross-section. The vector of covariates $x_{i,t}$ can include $1$ for the intercept, the lags of $y_{i,t}$, and the (higher-frequency observations of) covariates. 

\medskip

Our goal is to nowcast a set of variables $y_{i,t+h}$ with $i\in[N]$ at horizon $h$ measured at some low frequency, e.g.\ quarterly, $t\in[T]$, where $[p]=\{1,2,\dots,p\}$ for a positive integer $p$. For simplicity, we assume equally spaced data at different frequencies. In particular, $n_k^H$ denotes the total number of high-frequency observations for the $k^{\rm th}$ covariate for each low-frequency period $t$, and $n^L_k$ is the number of low-frequency periods used as lags.\footnote{For example, in our application a quarter of high-frequency lags used as covariates corresponds to $n^L_k=1$ while $n_k^H=3$ indicates that three months of data are used in each quarter. 
	Note that we can have  a mix of quarterly, monthly, and weekly data, and for each covariate indexed by $k$, $n_k^H$ represents different high-frequency sampling frequencies and associated lags $n_k^L n_k^H.$} The information set consists of $K$ predictors, i.e.\ 
$\left\{x_{i,t-j/n_k^H,k}:\;i\in[N],t\in[T],j=0,\dots,n^L_kn_k^H-1,k\in[K]\right\},$ measured potentially at some higher frequency, e.g., monthly/weekly/daily with real-time updates into the quarter.

\medskip

Consider the following panel data regression for the low-frequency panel target $y_{i,t|\tau},$ using information up to $\tau$:
\begin{equation*}
	y_{i,t|\tau} = \alpha + \sum_{q=1}^Q \rho_{i,q} y_{i,t-q} + \sum_{k=1}^{K}\psi(L^{1/n_k^H};\beta_k)x_{i,\tau,k} + u_{i,t|\tau},
\end{equation*}
where
\begin{equation}\label{eq:hf_lag}
	\psi(L^{1/n_k^H};\beta_k)x_{i,\tau,k} = \frac{1}{k_{\rm max}}\sum_{j=0}^{k_{\rm max}-1}\beta_{j,k}L^{j/n_k^H}x_{i,\tau-j/n_k^H,k},
\end{equation}
and	(a) $\alpha$ is the intercept constant across all $i$; (b) $\rho_{i,q}$ are autoregressive lag coefficients possibly different across $i$; (c) $k_{\rm max}$ is the maximum lag length which could depend on the covariate $k,$ and for each high-frequency covariate $x_{i,\tau,k}$ we have the most recent information at time $\tau.$\footnote{All parameters depend on $\tau$, but we suppress this detail to simplify notation. Furthermore, we will not keep track of the dependence of $x_t$, $k_{\max}$, and $K$ on $\tau$.} 
Note that this also implies the most recent vintage of data is used, i.e.\ past data revisions are incorporated as well. When $\tau$ $\leq$ $t - 1,$ we are dealing with a forecasting situation, hence, our analysis applies to forecasting as well.
Note also that $k_{\max}$ is the maximum lag length which may depend on the covariate $k,$ and for each high-frequency covariate $x_{i,\tau,k}$ we have the most up-to-date information available at time $\tau.$ For some high-frequency regressors that have not been updated yet, this could be stale information; see \cite{babii2022panelpe} for further discussion. In our empirical analysis, we consider the following two variations of this general formulation. For $i,j\in[N]$: (a) Pooled panel, denoted {\it pooled},  restriction: $\rho_{i,j} = \rho_j$. (b) Country-specific AR lag coefficients, denoted {\it HetAR}.
Note that in the HetAR case, we model country-specific effects, therefore we investigate whether more flexible models can improve the quality of nowcasts relative to the pooled panel model by including country-specific autoregressive coefficients in the model. Lastly, in the pooled, we group pooled autoregressive lags in time series dimension rather than the cross-section.

\medskip

The large number of predictors $K$ with potentially large number of high-frequency measurements $k_{\max}$ can be a rich source of predictive information, yet at the same time, estimating $N\times Q + 1 + \sum_{k=1}^K k_{\max}$ parameters, where $N$ is the size of the cross-section, $Q$ is number of autoregressive lags (constant across $i$) and $k_{\max}$ is the max of high-frequency lags for each covariate (constant across $i$), is  costly and may reduce the predictive performance in small samples. In addition, observing predictors at different frequencies leads to the frequency mismatch problem due to the missing data.
To solve both problems, we follow the MIDAS (machine learning) literature and instead of estimating a large number of individual slope coefficients, we consider a weight function\footnote{We use $q=1$ for nowcasting without low-frequency lags. More generally, $q-1$ denotes the number of low-frequency lags.} $\omega:[0,q]\to\R$, parameterized by $\beta_k\in\R^L$
\begin{equation*}
	\psi(L^{1/n_k^H};\beta_k)x_{i,\tau,k} = 
	\frac{1}{k_{\max}}\sum_{j=0}^{k_{\max}-1}\omega\left(\frac{j}{n_k^H};\beta_k\right)x_{i,\tau,k},
\end{equation*}
where $[0,q]\ni s\mapsto$ $\omega(s;\beta_k)$ = $\sum_{l=0}^{L-1}\beta_{l,k}w_l(s),$
and $(w_l)_{l\geq 0}$ is a set of $L$ functions, called the \textit{dictionary}. One can use either splines or Legendre polynomials as a dictionary; see \cite{babii2022machinepanel} for further discussion on different choices of the dictionary. The attractive feature of this approach is that we can map the MIDAS regression to the simple linear regression framework (as in Section \ref{sec:theory} of the paper). To that end, assuming that $k_{\max}$ is the same for all covariates and  $\mathbf{x}_i = (X_{i,1}W_k,\dots,X_{i,K}W_k)$, where for each $k\in[K]$, $X_{i,k} = \{x_{i,\tau-j/n_k^H,k},j = 0, \ldots, {k}_{\max} - 1\}_{\tau\in[T]}$ is a $T\times {k}_{\max}$ matrix of predictors and $W_k{k}_{\max} $ = $(w_l(j/n_k^H; \beta_k))_{0\leq l\leq L-1, 0 \leq j\leq {k}_{\max}}$ is a ${k}_{\max} \times L$ matrix corresponding to the dictionary. Note that the matrices $(W_k)_{k\in[K]}$ solve the frequency mismatch problem.

\subsection{A note on cross-validation}\label{cv.details}

Say the training and validation set for a given nowcast is $1,\dots,N$ in the cross-sectional dimension and $1,\dots, T_0$ in the time dimension. To select the tuning parameters of the sparse group LASSO, we use 5-fold cross-validation adapted to the panel setting. Specifically, the time dimension is partitioned into 5 contiguous folds, and for each fold the model is re-estimated on the remaining periods and validated on the excluded fold. The fold structure is replicated across all cross-sectional units so that entire time blocks, rather than individual observations, are left out. This ensures that dependence across countries at a given date is preserved and avoids information leakage across folds. The cross-validation error is then computed for each candidate penalty parameter, i.e, ($\gamma$, $\lambda$), and the optimal value is chosen for the minimum cross validation error as is standard in cross-validation implementations. We set the grid for the $\gamma = (0, 0.05, 0.1, \dots, 0.95, 1)$ while the grid for $\lambda$ is based on the usual path construction for the LASSO, i.e., for a given $\gamma$, we find the largest $\lambda$ such that the parameter vector is zero. We then construct the grid for $\lambda$ between the largest and the $\lambda$ value multiplied by a factor $10^{-2}$ in the log-space. Once the optimal combination of both tuning parameters are chosen, we re-estimate the model on the full data set using the optimal tuning parameters. We then use the estimated coefficients to predict the outcome at time $T_0+1$. We repeat the process by using expanding window in time dimension until all out-of-sample outcome data points are exhausted. 

Lastly, we note that the cross-validation procedure is implemented in the open-source R package \texttt{midasml} function \texttt{cv.panel.sglfit}. The same procedures where used in the articles \cite{babii2022machinepanel, babii2022panelpe} — see both articles for further details.

\section{Comparison with a Mixed Frequency  Vector Autoregressive model \label{appsec:MFVAR}}

Mixed-Frequency Vector AutoRegressive (MF-VAR) models have become valuable tools in the nowcasting framework, enabling the joint modelling of multiple time series observed at different frequencies while capturing the dynamic interdependencies among macroeconomic variables. Beyond improving nowcasting accuracy, MF-VAR models also support structural analysis through the interpretation of estimated coefficients, offering insights into the transmission mechanisms of economic shocks.
In a single-region context, two prominent contributions are \cite{schorfheide2015real}, who develop a Bayesian MF-VAR framework to forecast the U.S.\ business cycle using a mix of quarterly and monthly indicators, and \cite{chan2023high}, who extend this approach by incorporating weekly data. 
This standard setting has been further expanded by \cite{koop2020regional}, who apply the MF-VAR model to nowcast regional economic activity in the United Kingdom, and by \cite{koop2025monthly}, who focus on producing monthly growth estimates for individual U.S.\ states.

We include the MF-VAR as an additional benchmark in our analysis. Similar to \cite{schorfheide2015real} and \cite{koop2020regional}, we address the frequency mismatch between quarterly and monthly variables using the triangular decomposition method discussed in \cite{mariano2003new, mariano2010coincident}.
Given the inclusion of over 50 variables for each of the EA-19 countries and the relatively short time sample, the dataset is high-dimensional. Estimating the MF-VAR in this context poses challenges related to multicollinearity and the risk of overfitting.
Similar to \cite{schorfheide2015real}, we model GDP growth rates using a selected subset of monthly and quarterly indicators.\footnote{The MF-VAR is estimated using the following variables: GDP growth; Industrial Production (Manufacturing); Harmonised Index of Consumer Prices; Long-term Government Bond Yield; Unemployment Rate; Trade Balance; and Real Effective Exchange Rate.} 
We employ a Minnesota prior and incorporate stochastic volatility, conducting inference based on 10,000 replications.
The real-time nowcasting framework mirrors that of the main paper. We evaluate the performance of the MF-VAR by benchmarking it against an EA-19 model that relies solely on aggregate data.

Table \ref{tab:mfvar} reports the performance of two different specifications of the MF-VAR.
In Panel A, we fit the MF-VAR only on the EA-19 aggregate time series, while, in Panel B, we fit the MF-VAR on the individual country data and obtain the aggregate EA-19 forecast by using the weights $W^{(1)}-W^{(4)}$ defined in the main paper.
Looking at Panel A, we observe that the MF-VAR on aggregated data improves the nowcasting performance against the benchmark in all cases expect for the EoQ in the pre-COVID period.
Moving to Panel B, a similar pattern observed in the main paper with $W^{(3)}$ and $W^{(4)}$ achieving a better performance that the other weighting schemes. 
The MF-VAR on EA-19 aggregated data (Panel A) performs better than all specifications using country-level data (Panel B) in the Full sample, while we observe the opposite in the Pre-COVID period. 

All in all, the MF-VAR appears to provide reliable nowcasting performance, with gains that are generally robust across both the pre-COVID and full samples. 
While its effectiveness is confirmed when compared to the forecast combination results in Table 4 of the main paper, the linear combination panel MIDAS models generally outperform the MF-VAR across most horizons. The only exception is the 2-month-ahead nowcast in the full sample, where the MF-VAR yields a more accurate forecast. These findings highlight the MF-VAR’s competitiveness, particularly in specific forecasting windows, while also underscoring the advantages of combining forecasts in a panel MIDAS framework.

\begin{table}[]
	\centering
	\begin{tabular}{rccc c ccc}
        & \multicolumn{3}{c}{\textit{Full-sample}} && \multicolumn{3}{c}{\textit{Pre-COVID}} \\ 
        \cmidrule(lr){2-4} \cmidrule(lr){6-8}
        & 2-month & 1-month & EoQ && 2-month & 1-month & EoQ \\
        \hline	
        EA-19 & 4.960 & 3.642 & 3.366 && 0.185 & 0.215 & 0.176 \\
        & \multicolumn{6}{c}{Panel A. \it MF-VAR on EA-19 aggregated data}\\
        EA-19 & 0.656 & 0.893 & 0.977 && 0.948 & 0.860 & 1.077  \\ 
        & \multicolumn{6}{c}{Panel B. \it MF-VAR on individual country data}\\
        $W^{(1)}$ & 0.705 & 0.963 & 1.052 && 2.054 & 1.799 & 2.393 \\ 
        $W^{(2)}$ & 0.707 & 0.966 & 1.055 && 2.293 & 2.004 & 2.855 \\ 
        $W^{(3)}$ & 0.698 & 0.952 & 1.037 && 0.860 & 0.802 & 0.878 \\ 
        $W^{(4)}$ & 0.699 & 0.952 & 1.033 && 0.852 & 0.788 & 0.848 \\ 
        \hline
        \end{tabular}
	\medskip
	\caption{\small Nowcast comparison table --- root-mean-squared errors (RMSEs). 
		Nowcast horizons are 2- and 1-month ahead, as well as the end-of-the-quarter (EoQ). 
		The first row reports the absolute RMSE for the EA-19 model using only aggregate data. All remaining rows are relative RMSEs vis-\`a-vis the first row: :  values smaller than unity indicate an improvement in prediction with respect to the benchmark.
        Panel A displays results for MF-VAR on EU-19 aggregated data. 
        Panel B displays results for MF-VAR on individual country data and then aggregated at the Euro area level using weight $W^{(1)}$-$W^{(4)}$.
        Out-of-sample periods: 2016 Q1 - 2019 Q4 (pre-COVID) and 2016 Q1 - 2022 Q4 (Full sample).
		\label{tab:mfvar}} 
\end{table}


\newpage
\section{Variable Importance and Country-Level Evaluation \label{appsec:varimport}}

This section further investigates the nowcasting performance discussed in the main paper. It begins by exploring which variables drive the nowcast, and then evaluates the model’s performance at the country level to assess the added value of panel models without aggregating to the Euro area level.

\paragraph{Variable importance:}
Figure \ref{fig:sparsity_plots} displays the sparsity pattern resulting from the variable selection performed by the \textit{Pooled} panel model in the out-of-sample exercise. A colored tile indicates that the variable is selected by the sparse group penalty; a blank tile indicates otherwise. Overall, the selection pattern appears fairly persistent over time.

\begin{figure}
    \centering
    \includegraphics[scale=0.7]{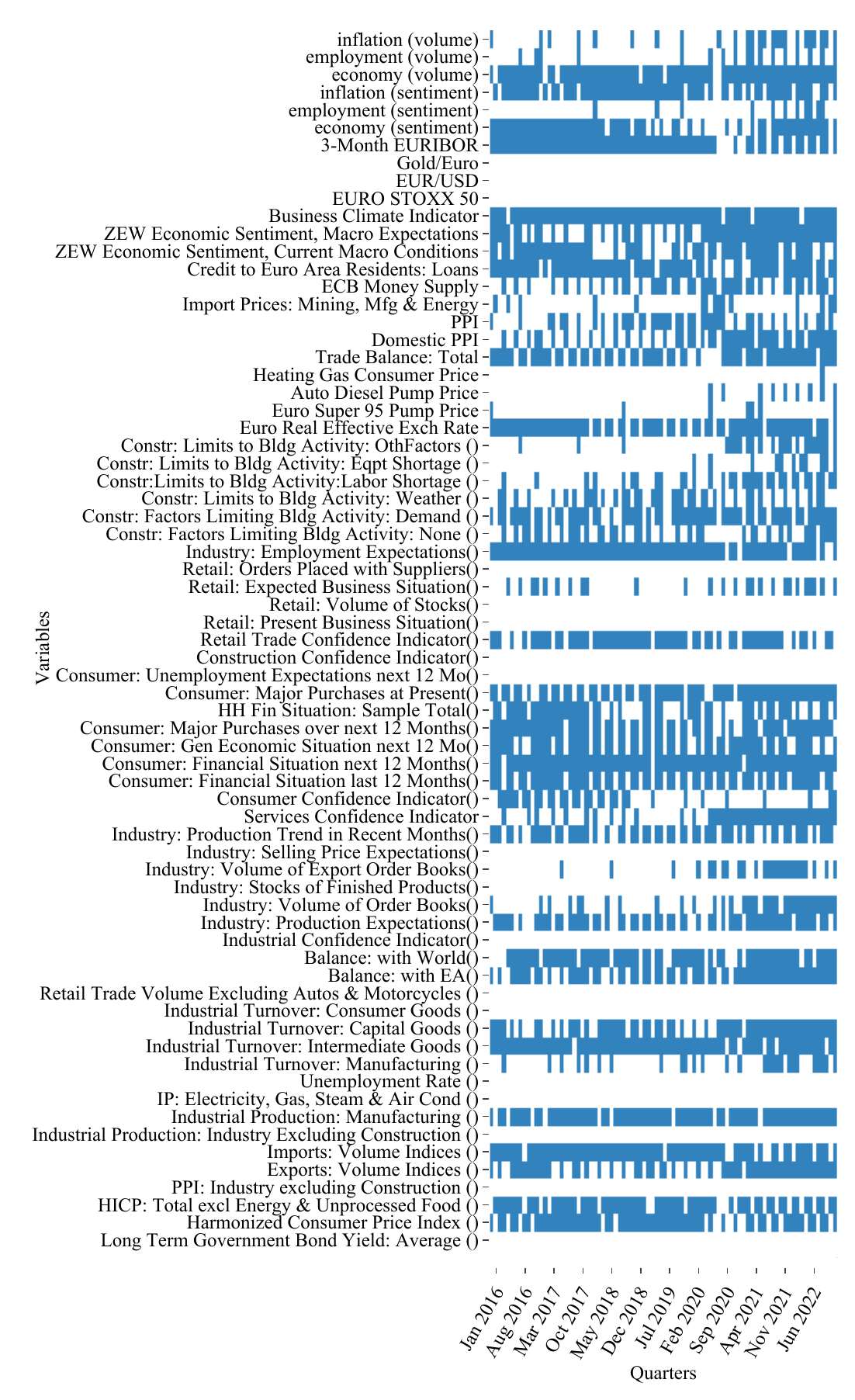}
    \caption{Sparsity plots of the variables selected by the \textit{Pooled} panel model in each month of the out-of-sample nowcasting exercise. A colored tile indicates that the variable is selected as an important predictor in that month.}
    \label{fig:sparsity_plots}
\end{figure}

Several noteworthy patterns emerge.  
First, news-based variables are frequently selected, whereas other high-frequency financial indicators—such as exchange rates and stock indices—are never chosen.  
Second, survey data generally provide useful signals. Among these, consumer-related surveys tend to be, on average, more informative than surveys on other topics (e.g., retail, construction, industry).  
Third, inflation-related variables (both news-based and HICP) are often selected. In contrast, energy- and fuel-related variables (e.g., Import Prices: Mining, Mfg \& Energy, gasoline, gas) are rarely picked, possibly because their information is already captured by HICP dynamics.  
Finally,  it is interesting to notice that labor-related variables are seldom selected (e.g., news-based employment indicators, unemployment rate).

Overall, these results are consistent with the main findings in the literature.  
First, we confirm the importance of soft data (e.g., surveys) in nowcasting Euro area economic developments, in line with \cite{cascaldigarcia2021}, who highlight that their timeliness can offset the long delays in the publication of hard data.  
Second, we emphasize the relevance of news-based variables in enhancing the predictive performance of nowcasting models tailored to the European context \cite{ashwin2021nowcasting, barbaglia2022forecastingEA}.  
Finally, our findings corroborate the relatively greater importance of inflation-related variables compared to unemployment indicators, consistent with \cite{boss2025nowcasting}.

\paragraph{Country-level performance:}
In Table \ref{tab:empirics_countrylevel_pooled}, we summarize the nowcasting performance at the country level. The table reports the RMSEs of the \textit{Pooled} panel model relative to those obtained from country-specific regressions. This comparison allows us to directly evaluate the contribution of the panel structure to nowcasting accuracy, independently of the weighting schemes used in the main analysis.\footnote{To conserve space, we present results only for the \textit{Pooled} panel model. Comparable findings are observed for the \textit{HetAR} panel model. Additional results are available upon request.}

Looking at the full sample (columns 2 to 4 of Table \ref{tab:empirics_countrylevel_pooled}), we observe that the panel model generally outperforms country-specific regressions. In over three-quarters of the cases, the relative RMSE falls below one, indicating superior nowcasting accuracy. On average, the panel model achieves an improvement of more than 10\% in forecasting performance. Notable gains are observed in large and medium-sized countries such as France, Italy, Spain, and Portugal.
However, the outlook changes when focusing on the pre-COVID sample (columns 5 to 7 of Table \ref{tab:empirics_countrylevel_pooled}). In this subsample, the panel model performs better in only about 40\% of the cases, with particularly weaker results in countries like France and Spain.

While this additional analysis provides valuable insights into the performance of panel models at the individual country level, it does not alter the main conclusions of the paper. On average, panel models deliver more accurate nowcasts when evaluated over the full sample, compared to the pre-COVID subsample. This finding holds both at the country level and when aggregating to the Euro area, suggesting that the influence of the weighting schemes used in the main analysis is only marginal.

\begin{table}[ht]
\centering
\begin{tabular}{c ccc c ccc}
  \hline
country & 2-month & 1-month & EoQ && 2-month & 1-month & EoQ \\ 

		& \multicolumn{3}{c}{\it Full sample} &&\multicolumn{3}{c}{\it Pre-COVID}\\[0.2em]

  \hline
AT & 1.03 & 0.82 & 0.77 && 0.80 & 1.75 & 1.65 \\ 
  BE & 0.98 & 0.93 & 0.96 && 1.90 & 1.73 & 2.03 \\ 
  CY & 1.07 & 0.70 & 0.90 && 0.79 & 1.45 & 1.13 \\ 
  DE & 1.16 & 1.06 & 0.92 && 0.84 & 1.36 & 1.33 \\ 
  EE & 0.90 & 0.98 & 1.17 && 1.04 & 0.89 & 0.81 \\ 
  ES & 0.76 & 0.51 & 0.50 && 1.64 & 3.77 & 2.06 \\ 
  FI & 0.88 & 0.94 & 1.89 && 1.07 & 1.45 & 1.25 \\ 
  FR & 0.68 & 0.62 & 0.74 && 1.54 & 3.13 & 2.50 \\ 
  GR & 0.91 & 1.01 & 1.03 && 1.17 & 0.96 & 1.23 \\ 
  IE & 0.79 & 0.95 & 0.97 && 0.69 & 0.95 & 0.95 \\ 
  IT & 0.89 & 0.64 & 0.76 && 1.78 & 1.85 & 1.92 \\ 
  LT & 0.70 & 1.09 & 1.04 && 0.69 & 0.88 & 0.64 \\ 
  LU & 1.10 & 1.03 & 0.80 && 0.77 & 0.91 & 1.14 \\ 
  LV & 0.92 & 1.05 & 0.75 && 0.85 & 0.92 & 1.47 \\ 
  MT & 0.88 & 0.91 & 0.92 && 0.99 & 1.09 & 1.31 \\ 
  NL & 0.97 & 0.90 & 0.98 && 1.02 & 1.65 & 1.32 \\ 
  PT & 0.91 & 0.85 & 0.65 && 1.28 & 0.79 & 0.93 \\ 
  SI & 0.80 & 0.90 & 0.72 && 0.76 & 0.76 & 0.54 \\ 
  SK & 1.04 & 0.69 & 0.64 && 0.91 & 0.78 & 1.57 \\ 
   \hline
\end{tabular}

	\medskip
	\caption{\small Nowcast comparison table by country --- root-mean-squared errors (RMSEs) of the \textit{Pooled} panel model relative to \textit{country-specific regressions}. 
    Nowcast horizons are 2- and 1-month ahead, as well as the end-of-the-quarter (EoQ).  
    We use 5-fold cross-validation adjusted for panel data to compute both sg-LASSO tuning parameters.
    Out-of-sample periods: 2016 Q1 - 2019 Q4 (pre-COVID) and 2016 Q1 - 2022 Q4 (Full sample).
		\label{tab:empirics_countrylevel_pooled}} 

\end{table}

\newpage

\section{Further Digression on Weights \label{appsec:weights}}

To further investigate the nowcasting performance across weighting schemes, we turn to Figure \ref{fig:weights_W3W4} where we compare the weights obtained by  $W^{(4)}$ (i.e., projections on GDP) and by $W^{(3)}$ (i.e., proportion of GDP level).
We take the difference between  $W^{(4)}$ and  $W^{(3)}$ weights: positive values indicate a larger weight given by $W^{(4)}$ with respect to $W^{(3)}$.
Note that by construction the $W^{(3)}$ weights are a direct measure of the size of each national economy within the Euro area. 
Compared to $W^{(3)}$, we observe that $W^{(4)}$ assigns smaller weights to Germany and to a lesser extent the Netherlands, while it inflates the relative importance of some small- and medium-sized countries, namely Austria, Belgium and Luxembourg.
Although the  size of these economies within the Euro area is relatively small,  information about economic developments in those countries play a relevant role in attaining more accurate nowcasts.

\begin{figure}[]
	\centering
	\makebox[\textwidth][c]{\includegraphics[scale=0.6]{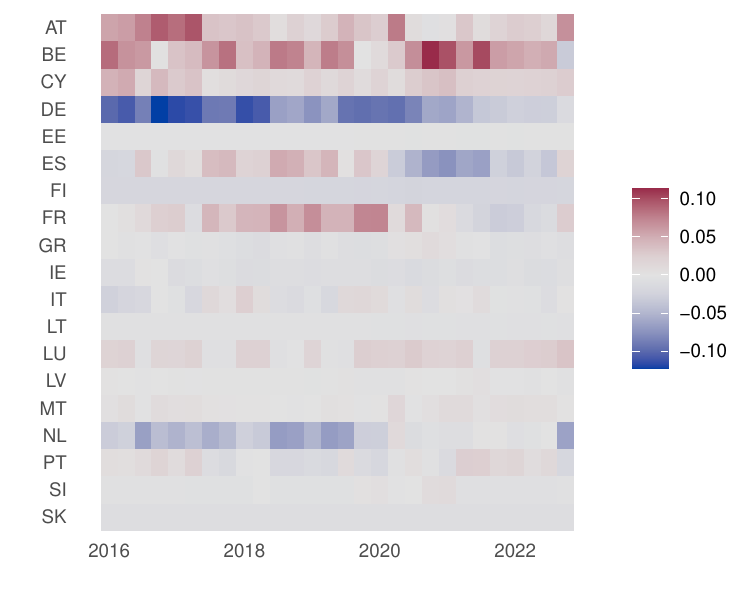}}%
	\caption{Difference between $W^{(4)}$ and $W^{(3)}$ weights by country. 
		If difference is larger (smaller) than zero, $W^{(4)}$ weights are larger (smaller) than $W^{(3)}$ ones and are reported in red (blue).      
		Each tile corresponds to a quarter. 
		The darker the color, the larger the difference between weighting schemes.}
	\label{fig:weights_W3W4}
\end{figure}
\clearpage

\end{document}